\documentclass[journal,10pt]{IEEEtran}
\usepackage{graphicx,cite}
\usepackage{multirow,amsmath, mathrsfs,soul,ulem}
\usepackage{algorithm,algorithmic,rotating,url}
\usepackage{fancybox,subfigure,amsthm,afterpage}
\usepackage{times,amsmath,bm,amssymb,color,setspace,multirow,stfloats,extpfeil}
\definecolor{orange}{cmyk}{0,0.5,1,0}
\definecolor{green}{cmyk}{1,0.4,.8,0}
\definecolor{blue}{rgb}{0.2,0.3,0.8}
\definecolor{red}{rgb}{0.8,0.1,0.1}

\newcommand{\thdc}[3]{\;\raisebox{-2ex}{\shortstack[c]{\ensuremath{#1}\footnotesize{=#2} \\  \ensuremath{\gtrless} \\ \ensuremath{#1}\footnotesize{=#3}}} \;}

\def \ba {\bm {a}}

\newcommand{\bcP}{\bm{{\cal P}}}
\newcommand{\bpd}{\bm{p}_{\bm{d}}}

\newcommand{\bpf}{\bm{p}_{\bm{f}}}
\newcommand{\Var}{\mbox{Var}}

\newcommand{\cN} {{\cal N}}
\newcommand{\cO} {{\cal O}}

\newcommand{\cP} {{\cal P}}
\newcommand{\cH} {{\cal H}}
\newcommand{\bI}{\bm{I}}
\newcommand{\bA}{\bm{A}}

\newcommand{\bD}{\bm{D}}

\newcommand{\bK}{\bm{K}}
\newcommand{\bH}{\bm{H}}

\newcommand{\bg}{\bm{g}}
\newcommand{\bs}{\bm{s}}

\newcommand{\bQ}{\bm{Q}}
\newcommand{\bq}{\bm{q}}
\newcommand{\by}{\bm{y}}
\newcommand{\bh}{\bm{h}}
\newcommand{\bx}{\bm{x}}
\newcommand{\bn}{\bm{n}}

\newcommand{\bb}{\bm{b}}

\newcommand{\bou}{\bm{u}}
\newcommand{\bow}{\bm{w}}
\newcommand{\bz}{\bm{z}}

\newcommand{\bSigma}{\bm{\Sigma}}
\newcommand{\bLa}{\bm{\Lambda}}
\newcommand{\bphi}{\bm{\phi}}
\newcommand{\bpsi}{\bm{\psi}}

\newcommand{\bmu}{\bm{\mu}}

\newcommand{\bTh}{\bm{\Theta}}
\newcommand{\bone}{\bm{1}}
\newcommand{\bzero}{\bm{0}}
\newcommand{\mE}{\mathbb E}
\newcommand{\mP}{\mathbb P}

\newcommand {\DIAG}{\mbox{DIAG}}

\newcommand {\MDC}{\mbox{MDC}}

\newcommand {\Rel}{\mbox{Re}}

\newcommand {\mW}{\mbox{mW}}

\def\Var{\mbox{Var}}

\newtheorem{lem}{Lemma}

\title{On Power Allocation for Distributed Detection with Correlated Observations and Linear Fusion \thanks{This work is supported by the National Science Foundation under grants CCF-1341966 and CCF-1319770.}}

\author{\normalsize Hamid R. Ahmadi, ~\IEEEmembership{Member,~IEEE,} Nahal Maleki, ~\IEEEmembership{Student Member,~IEEE,}  Azadeh Vosoughi,~\IEEEmembership{Senior Member,~IEEE} \\
}
\begin{document}
	
\maketitle

\vspace{-2.3cm}
%\begin{abstract}
%===========================================================================================
%
\begin{abstract}
We consider a binary hypothesis testing problem in an inhomogeneous wireless sensor network, where a fusion center (FC) makes a global decision on the underlying hypothesis. We assume sensors' observations are correlated Gaussian and sensors are unaware of this correlation when making decisions. Sensors send their modulated decisions over fading channels, subject to individual and/or total transmit power constraints. For parallel-access channel (PAC) and multiple-access channel (MAC) models, we derive modified deflection coefficient (MDC) of the test statistic at the FC with coherent reception. We propose a transmit power allocation scheme, which maximizes MDC of  the test statistic, under three different sets of transmit power constraints: total power constraint, individual and total power constraints, individual power constraints only. When analytical solutions to our constrained optimization problems are elusive, we discuss how these problems can be converted to convex ones. We study how correlation among sensors' observations, reliability of local decisions, communication channel model and channel qualities and transmit power constraints affect the reliability of the global decision and power allocation of inhomogeneous sensors.
\end{abstract}
\begin{IEEEkeywords}
%\vspace{-0.1cm}
Distributed detection, coherent reception, modified deflection coefficient, power allocation, correlated observations, linear fusion, parallel-access channel, multiple-access channel.
\end{IEEEkeywords}
\vspace{-0.5cm}
\section{Introduction}
\vspace{-0.1cm}
The classical problem of binary distributed detection in a network consisting of multiple distributed sensors and a fusion center (FC), has a long and rich history. Each sensor (local detector)
processes its single observation locally and passes its binary decision to the FC, that is tasked with fusing the binary decisions received from the individual sensors and deciding which of the two underlying hypotheses is true \cite{Varshney_Book, Viswanathan_PartI, yan2001distributed}. Motivated by the potential application of wireless sensor networks (WSNs) for event monitoring, researchers have further studied this problem and extended its setup, taking into account that bandwidth-constrained communication channels between sensors and the FC are error-prone, due to limited transmit power to combat noise and fading (so-called channel aware binary distributed detection \cite{Magazine-Paper, Perfect-CSI, Statistics-CSI}). Given each sensor makes its binary decision based on one local observation, they have investigated how the reliability of the final decision  at the FC is affected by performance indices of local detectors (sensors) as well as wireless channel properties. Following these works, we consider channel aware binary distributed detection in a WSN with coherent reception at the FC \cite{Ahmadi_Vosoughi_SPL, Hamid-topo, Hamid-uncertainty}. In this paper, our goal is to study transmit power allocation, when each sensor has an individual transmit power constraint and/or all sensors have a joint transmit power constraint, such that the reliability of the final decision at the FC is maximized.

%===========power allocation in Distributed detection===============

Power allocation for channel aware binary distributed detection in WSNs has been studied in \cite{Poor_PA, Kim-Det}. More specifically, \cite{Poor_PA} studied the power allocation that maximizes the J-divergence between the distributions of the received signals at the FC under two different hypotheses, subject to individual and total transmit power constraints on the sensors, with parallel access channel (PAC) \footnote{In PAC, channels between the sensors and the FC are orthogonal (non-interfering). This can be realized by either time, frequency, or code division multiple access \cite{PAC-MAC, Xiao2008}.} and coherent reception at the FC (i.e., channel phases are known and compensated at the sensors). Leveraging on \cite{Poor_PA}, \cite{Kim-Det} studied detection outage and detection diversity, as the number of sensors goes to infinity, and sensors have identical performance indices. Note that \cite{Poor_PA, Kim-Det} assume the sensors have uncorrelated observations under each hypothesis.

%=========power allocation in Disributed Estimation=========

Power allocation in WSNs has also been studied for distributed estimation
\cite{Hafeez, Cui, Fang, wu2013power, alirezaei2015optimum, Varshney-Est-spatial, Varshney-Est-spatial-cor, behbahani2012linear, bahceci2008linear, Xiao2008}, where some works minimized the mean square error (MSE) of an estimator subject to certain transmit power constraints \cite{Hafeez, Cui, Fang, wu2013power, alirezaei2015optimum, Varshney-Est-spatial, Varshney-Est-spatial-cor, behbahani2012linear}, while others minimized total transmit power subject to a constraint on the MSE of an estimator \cite{ bahceci2008linear,Xiao2008}. These works, except \cite{Varshney-Est-spatial, Varshney-Est-spatial-cor, Xiao2008}, mainly focus on PAC with coherent reception at the FC. In \cite{Varshney-Est-spatial, Varshney-Est-spatial-cor, Xiao2008}, sensors and the FC are connected differently via a multiple-access channel (MAC), where the individual sensors send their signals simultaneously, albeit after channel phases are compensated at the sensors, and the FC receives the coherent sum of these transmitted signals. Most of these works assume the sensors' observations are uncorrelated, with the exception of \cite{Hafeez, Fang, bahceci2008linear}. In \cite{Varshney-Est-spatial, Varshney-Est-spatial-cor} sensors collaborate with each other by linearly combining their independent observations before sending to the FC.

%========PAC and MAC in Distributed detection=====================

For binary distributed detection in WSNs, \cite{PAC-MAC} compared the detection performance using both PAC and MAC, with linear fusion rule and noncoherent reception at the FC (i.e., no channel phase compensation at the sensors), albeit without imposing any transmit power constraint. Assuming the sensors' observations are uncorrelated under each hypothesis and the FC utilizes a linear fusion rule when using PAC, \cite{PAC-MAC}  showed that coherent MAC outperforms coherent PAC, whereas noncoherent PAC (MAC) outperforms noncoherent MAC (PAC) when sensors' decisions are (un)reliable.
%
%============correlation in Distributed detection==========
%
Distributed detection with correlated observations has been studied
assuming error-free \cite{yan2001distributed,good_bad,chen2012new} and erroneous communication channels \cite{lin2015distributed}. The focus of these works though is on how to design optimal local and global decisions rules to improve the detection reliability at the FC, assuming sensors know the correlation among their observations.
Different from \cite{ yan2001distributed,good_bad,chen2012new,lin2015distributed} we focus on how to optimally transmit the sensors' decisions to the FC within certain transmit power constraints, with a linear fusion rule at the FC and assuming sensors are unaware of the correlation among their observations.
%they mainly focus on finding optimal sensor rule in distributed detection which considering correlated observations, they show that it is not likelihood ratio test (LRT)-based. None of the works in the field of distributed detection considered correlation between sensors' observations and its effect on the power allocation.

%======our contributions==========================================
\vspace{-0.1cm}
{\bf Our Contributions}: We consider a binary hypothesis testing problem using $M$ sensors and a FC, where under ${\cal H}_0$, sensors' observations are uncorrelated Gaussian with covariance matrix $\sigma_0^2\bI$ and under ${\cal H}_1$ they are correlated Gaussian with a non-diagonal covariance matrix $\Sigma$. We relax the assumption in \cite{yan2001distributed,good_bad,chen2012new,lin2015distributed} that sensors know the correlation among their observations and consider a more practical scenario, where the sensors are unaware of such correlation. Sensors send their modulated binary decisions over nonideal fading channels, subject to individual and/or total transmit power constraints. We consider PAC and MAC with coherent reception at the FC, assuming that channel phases are compensated at the sensors similar to \cite{Varshney-Est-spatial, Varshney-Est-spatial-cor,Xiao2008}. To curb the hardware and computational complexity and also have a fair comparison between PAC and MAC, we assume that, when the sensors and the FC are connected via PAC, the FC utilizes a linear fusion rule to obtain the global test statistic $T$. We propose a transmit power allocation scheme, which maximizes modified deflection coefficient (MDC) of  $T$. We choose MDC as the performance metric, since unlike detection probability and J-divergence that require the probability distribution function of $T$, obtaining MDC only needs the first and second order statistics of $T$, and often renders a closed-form expression \cite{Sayed_Linear,picin-df,VVV}. Also, an MDC-based optimization problem can lead into near-optimal solutions for its corresponding detection probability-based optimization problem with much less computational complexity\cite{Sayed_Linear,lin_fading,Ahmadi_Vosoughi_SPL}.
We obtain the MDC of $T$ for coherent PAC and MAC in closed-forms that depends on the correlation among sensors' observations. Considering three different sets of transmit power constraints, we investigate transmit power allocation schemes that maximize the MDC. Under the conditions that analytical solutions to our constrained optimization problems are elusive, we discuss how these problems can be converted to convex ones and thus can be solved numerically.

%========paper organization===========
%\vspace{-0.1cm}
{\bf Paper Organization}:  Section \ref{system-model} details our system model and three different sets of transmit power constraints. Section \ref{derive DF} derives the MDC of $T$ for coherent PAC and MAC in closed-form expressions. Section \ref{power-alloc} formulates three different sets of constrained optimization problems and describes our approach to solve these problems. Section \ref{num-results} presents our numerical results for different correlation values, sensors' observations and communication channel qualities. Section \ref{conc} concludes the paper.

%==========notation============================
%\vspace{-0.1cm}
{\bf Notations}: Scalars, vectors and matrices are denoted by non-boldface lower, boldface lower, and boldface upper case letters, respectively. A Gaussian random vector $\bx$ with mean vector $\bmu$ and covariance matrix $\bSigma$ is shown as $\bx \sim \cN(\bmu,\bSigma)$. Transpose and complex conjugate transpose (Hermitian) of vector $\ba$ are denoted as $\ba^T$ and $\ba^H$, respectively. $\DIAG\{\ba\}$ represents a diagonal matrix
whose diagonal elements are the components
of column vector $\ba$. $\bA\succ 0$ ($\bA\succeq 0$ ) indicates that $\bA$ is a positive (semi-)definite matrix. $\ba \succ \bb$ ($\ba \succeq \bb$) indicates that each entry of $\ba$ is greater than (or equal to) the corresponding entry of $\bb$. $\Rel\{x\}$ is the real part of $x$. $\bzero\!=\![0,...,0]^T$ and $\bone \!=\! [1,...,1]^T$ are two $M \times 1$ vectors.
% and $\bI$ is an identity matrix of size $M$.
The $(i,j)$ entry of matrix $\bA$ is indicated with $[\bA]_{ij}$.
%$\otimes$ is the Kronecker product.
For vector $\ba$ we have $||\ba||^2=\ba^T\ba$ and $||\ba||=\sqrt{\ba^T\ba}$.
%
%==================================================================
%
\vspace{-0.8cm}
\section{System Model and Problem Statement}\label{system-model}
\vspace{-0.4cm}
%We consider the problem of testing two hypotheses ${\cal H}_0$ and ${\cal H}_1$.
Our system model consists of an FC and $M$ distributed sensors with observation vector $\bx= [x_1,x_2,...,x_M]^T$. The FC is tasked with solving the binary hypothesis testing problem
${\cal H}_0 : \bx \sim {\cal N}(0,\sigma_0\bI), {\cal H}_1 : \bx \sim {\cal N}(0,\bSigma)$,
where $\sigma_0$ is the variance under ${\cal H}_0$ and $\bSigma$ is a non-diagonal covariance matrix under ${\cal H}_1$ with diagonal entries different from $\sigma_0$, i.e., under  ${\cal H}_1$ (${\cal H}_0$) sensors' observations are correlated (uncorrelated) Gaussian variables with different energy levels. Suppose sensor $k$, only based on its own observation $x_k$, makes a binary decision \cite{Poor_PA,Perfect-CSI,PAC-MAC} and maps it to $u_k=1$ ($u_k=0$) when it decides ${\cal H}_1$ (${\cal H}_0$), i.e., we assume that sensor $k$ is unaware of the correlation among sensors' observations, $\bSigma$. We denote $p_{f_k}\!=\!\mP(u_k=1|{\cal H}_0)$ and $p_{d_k}\!=\!\mP(u_k=1|{\cal H}_1)$ as the false alarm and detection probabilities of sensor $k$ and assume $p_{d_k} > p_{f_k}$.
The decision $u_k$ is communicated to the FC over a fading channel with transmit power ${\cal P}_{t_k}$. Let $h_k = |h_k|e^{j\varphi_k}$ denote the complex fading coefficient corresponding to sensor $k$, with $|h_k|$ and $\varphi_k$ being the channel amplitude and phase, respectively. Let $y_k$ and $y$ denote the channel output corresponding to the channel input $u_k$ and ($u_1, u_2,...,u_M$), when the sensors and the FC are connected via PAC and MAC, respectively. Since channel phases are compensated at the sensors, we have \cite{PAC-MAC}
%
%\begin{align}\label{y-k}
%&\mbox{coherent~PAC}:~y_{k}= \sqrt{{\cal P}_{k}}|h_k| u_k + n_{k}, &&\quad \mbox{noncoherent~PAC}:~y_{k}= \sqrt{{\cal P}_{k}}h_k u_k + n_{k},~k=1,...,M\nonumber\\
%&\mbox{coherent~MAC}:~y= \sum_{k=1}^M\sqrt{{\cal P}_{k}} |h_k| u_k + n,
% &&\quad
%\mbox{noncoherent~MAC}:~y= \sum_{k=1}^M\sqrt{{\cal P}_{k}} h_k u_k + n,
%\end{align}
%
\begin{align}\label{y-k}
&\mbox{PAC}:~y_{k}= \sqrt{{\cal P}_{k}}|h_k| u_k + n_{k},~k=1,...,M \mbox{~~and}\nonumber\\
&\mbox{MAC}:~y= \sum_{k=1}^M\sqrt{{\cal P}_{k}} |h_k| u_k + n
\end{align}
where communication channel noises are $n_{k} \! \sim \! {\cal C}{\cal N}(0,\sigma_{n}^2)$ and $n \! \sim \! {\cal C}{\cal N}(0,\sigma_{n}^2)$. Fading coefficients $h_k$'s and noises $n_{k}$'s and $n$ are all mutually uncorrelated and $h_k$ is assumed to be constant during a detection interval.
%$\mathbb{E}\{h_k\}\!= \! 0$ and $\mathbb{E}\{|h_k|^2\}\! = \! 1$.
Also ${\cal P}_{k} = {\cal P}_{t_k} \theta_k$, where $\theta_k = G d_{FS_k}^{-\epsilon_c}$, $d_{FS_k}$ is the distance between sensor $k$ and the FC, $\epsilon_c$ is the pathloss exponent, and $G$ is a constant.
We assume that the FC obtains a test statistic $T$ from the channel output(s) and makes a global decision $u_0 \in  \{0,1\}$ where $u_0\!=\!1$ and $u_0\!=\!0$ correspond to ${\cal H}_1$ and ${\cal H}_0$, respectively. In particular,
the FC applies $T \thdc{u_0}{1}{0}\tau_0$
where the threshold $\tau_0$ is chosen to maximize  the total detection probability $P_{D_0}\!=\!\mP(u_0=1|{\cal H}_1)$ at the FC, subject to the constraint that the total false alarm probability satisfies $P_{F_0}\!=\!\mP(u_0=1|{\cal H}_0) \leq \beta_F$ at the FC, where $\beta_F \in (0,1)$. In a PAC, we assume that the FC is restricted to utilize a linear fusion rule to obtain the test statistic $T$ \cite{Perfect-CSI,PAC-MAC}. Implementing the linear fusion rule has low complexity and allows a fair comparison between PAC and MAC. Furthermore, the authors in \cite{Perfect-CSI} have shown that, when identical sensors and the FC are connected via PAC, the linear fusion rule is a good approximation to the optimal Likelihood Ratio Test (LRT) rule at low signal-to-noise-ratio (SNR) regime.
We let $T$ be\\
%\vspace{-0.5cm}
%
\begin{align}\label{fusion-statistic}
\mbox{PAC}:~T=\sum_{k=1}^M  \Rel(y_k),\quad \mbox{~MAC}:~T= \Rel(y).
\end{align}
We consider coherent PAC and MAC with channel phase compensation at the sensors \cite{Poor_PA, Perfect-CSI}.
Our goal is to find the transmit powers at sensors such that the MDC of  $T$ is maximized, subject to different sets of power constraints. We refer to these as the {\it MDC-based transmit power allocation}. We consider three different sets of transmit power constraints: ($A$) there is a total power constraint (TPC) such that $\sum_{k=1}^{M}{\cal P}_{t_k} \leq{\cal P}_{tot}$, where ${\cal P}_{tot}$ is the total transmit power budget among sensors,  we refer to this set as TPC; ($B$) there is an individual power constraint  (IPC) for each sensor such that $0\leq {\cal P}_{t_k} \leq{\cal P}_{0_k}$ as well as a TPC $\sum_{k=1}^{M}{\cal P}_{t_k} \leq {\cal P}_{tot}$, where ${\cal P}_{tot} < \sum_{k=1}^M {\cal P}_{0_k}$, we refer to this set as TIPC; ($C$) there are only IPCs for sensors such that $0\leq {\cal P}_{t_k} \leq{\cal P}_{0_k}$,  we refer to this set as IPC.\

Section \ref{derive DF} drives the MDC of $T$ for coherent PAC and MAC. The MDC-based transmit power allocations under these three different sets of power constraints are discussed in Section \ref{power-alloc}.

%=====================================================================
\vspace{-0.2cm}
\section{Deriving Modified Deflection Coefficient}\label{derive DF}
\vspace{-0.2cm}
Before delving into the derivations, we introduce the following definitions and notations. Consider the signal model in (\ref{y-k}) and (\ref{fusion-statistic}). We let $a_k\!=\! \sqrt{{\cal P}_{k}}$, $w_k \! = \! \Rel(n_k)$,  $w \!=\!\Rel(n)$. We define the column vectors $\bh\!=\![h_1,...,h_M]^T$, $|\bh|\!=\![|h_1|,...,|h_M|]^T$, $\by\!=\![y_1,...,y_M]^T$, $\ba\!=\![a_1,...,a_M]^T$, $\bow \!=\! [w_1,...,w_M]^T$, $\bn \!=\! [n_1,...,n_M]^T$, $\bpd \!=\! [p_{d_1}, ..., p_{d_M}]^T$, $\bpf \!= \! [p_{f_1}, ..., p_{f_M}]^T$, $\bou\!= \![u_1, u_2,...,u_M]^T$, $\bcP \!=\! [\cP_1,..., \cP_M]^T$, $\bpsi\!=\![\psi_1,...,\psi_M]$, $\bphi\!=\![\phi_1,...,\phi_M]$, and the square matrix
$|\bH| \!= \! \DIAG\{|\bh|\}$.\

%
%================Coherent PAC and MAC=======================================
We define the MDC of $T$ as \cite{Sayed_Linear}
\begin{equation}\label{DF-conditional}
 \MDC = \frac{\big( \mathbb{E}\{T|{\cal H}_1, \bh\} - \mathbb{E}\{T|{\cal H}_0, \bh\}\big)^2}{\Var\{T|{\cal H}_1, \bh\}},
\end{equation}
where $\mathbb{E}\{.\}$ and $\Var\{.\}$ are performed with respect to the channel inputs $u_k$'s and the channel noises.
%for a given a channel realization $\bh$.
To calculate $\mE\{T|\cH_i, \bh\}$ for $i\!=\!0,1$ and $\Var\{T|\cH_1, \bh\}$ in (\ref{DF-conditional}), we use the Bayes rule and the fact that $\cH_i\rightarrow u_k \rightarrow y_k(y) \rightarrow u_0$ in PAC(MAC) form Markov chains for $i\!=\!0,1$. Hence
\begin{eqnarray}
 \mE\{T|\cH_i, \bh\} \!\!&\!\! =\!\!&\!\!  \sum_{\bou}\mE\{T|\bou, \bh\} \mP(\bou|\cH_i),~~~~i=0,1\label{mean-csi}\\
\Var\{T|\cH_1 , \bh\} \!\!&\!\! =\!\!&\!\!  \sum_{\bou} \bar{\Delta} \mP(\bou|\cH_1), ~\\
\mbox{where}~\bar{\Delta}&=&\mE\left\{\left(T- \mE\{T|\cH_1, \bh\}\right)^2|\bou,\bh\right\}, \label{var0-csi}\nonumber
\end{eqnarray}
and the sums are taken over all values of vector $\bou$. To simplify $\bar{\Delta}$ in (\ref{var0-csi}), we add and subtract $\mE\{T|\bou, \bh\}$ to the terms inside the  parenthesis in (\ref{var0-csi}) and expand the products. We have
\begin{align}\label{var-csi}
&\bar{\Delta} =\\
&\underbrace{\mE\left\{(T - \mE\{T|\bou, \bh\})^2|\bou\right\}}_{\Delta} + \underbrace{(\mE\{T|\bou, \bh\} - \mE\{T|\cH_1, \bh\})^2}_{\Delta'}\nonumber\\
&+ 2 \mE\left\{ (T- \mE\{ T |\bou, \bh\})(\mE\{ T |\bou, \bh\} - \mE\{ T | \cH_1, \bh\}) |\bou \right\}.\nonumber
\end{align}
We observe that the last term in (\ref{var-csi}) is zero. Thus $\bar{\Delta}$ in (\ref{var-csi}) is simplified to $\bar{\Delta}=\Delta+\Delta'$. Using (\ref{mean-csi}), (\ref{var0-csi}) and (\ref{var-csi}), we derive the MDC in the following.
%\vspace{-1cm}
%========cohernt PAC=============================
\subsection{\textbf{PAC}} Considering the signal model in (\ref{y-k}) and (\ref{fusion-statistic}), we have $\Rel(y_k)\!=\!   a_k |h_k|u_k + w_k$ where $ w_k \! \sim \! {\cal N}(0,\frac{\sigma_{n}^2}{2})$. We write  $T \!= \!\ba^T |\bH| \bou + \bone^T \bow$. Therefore
%given $\bou$ and $\bh$, we find
$\mE\{T |\bou, \bh\}\!=\! \ba^T |\bH| \bou$. %and variance $\Var\{ T |\bou, \bh\}\!=\! M\frac{r\sigma_{n}^2}{2}$.
Substituting $\mE\{T |\bou, \bh\}$ into \eqref{mean-csi} and using the facts $\bpd \!=\! \mE\{\bou|\cH_1 \} \!=\!\sum_{\bou} \bou \mP(\bou|\cH_1)$ and $\bpf \!=\! \mE\{\bou|\cH_0 \} \!=\!\sum_{\bou} \bou \mP(\bou|\cH_0)$ we find
\begin{equation}\label{final-mean-case-A}
\mE\{T|\cH_1, \bh\} \!=\! \ba^T |\bH| \bpd,~~~\mbox{and}~~~ \mE\{T|\cH_0, \bh\} \!= \! \ba^T |\bH| \bpf.
\end{equation}
Next, we derive $\Delta$, $\Delta'$ for $\Var\{T|\cH_1 , \bh\}$. Since $T - \mE\{T|\bou, \bh\}=\bone^T \bow$, we find $\Delta \!= \! \mE\{\bone^T \bow \bow^T \bone\} = M\frac{\sigma_{n}^2}{2}$. Also, because $\mE\{T|\bou, \bh\} - \mE\{T|\cH_1, \bh\} \!=\! \ba^T |\bH| (\bou-\bpd) $, we have $\Delta' \!=\!  \ba^T |\bH|(\bou- \bpd)(\bou- \bpd)^T |\bH| \ba $. Substituting $\bar{\Delta} =\Delta+\Delta'$ into (\ref{var0-csi}) and using the facts $\sum_{\bou} \mP(\bou|\cH_1)=1$, $\sum_{\bou} (\bou- \bpd)(\bou- \bpd)^T \mP(\bou|\cH_1)=\mE\{\bou\bou^T|\cH_1\}- \bpd \bpd^T$, we reach to
\begin{equation}\label{final-var-case-A}
\Var\{T|\cH_1,\bh\}
= M\frac{\sigma_{n}^2}{2} +  \ba^T |\bH|(\bar{\bm{P}}_{\bm{d}} - \bpd\bpd^T)|\bH|\ba,
\end{equation}
where $\bar{\bm{P}}_{\bm{d}} \!= \! \mE\{\bou\bou^T|\cH_1\}$ is a square matrix with diagonal entries $[\bar{\bm{P}}_{\bm{d}}]_{ii}\!=\!p_{d_i}$ and off-diagonal entries $[\bar{\bm{P}}_{\bm{d}}]_{ij}\!=\!\mP(u_i\!=\!1, u_j\!=\!1|\cH_1)$ for $i,j \!=\! 1,...,M, i\neq j$.
Note that the correlation among sensors' observations affects the off-diagonal entries of $\bar{\bm{P}}_{\bm{d}}$, i.e., for independent observations $[\bar{\bm{P}}_{\bm{d}}]_{ij}\!=\!p_{d_i}p_{d_j}$ for all $i \neq j$ and equivalently
%for fully correlated observations $p_{d_i}\!=\! p_d$ for all $i$  and thus $[\bar{\bm{P}}_{\bm{d}}]_{ij}\!=\! p_d$ for all $i,j$. Equivalently
%
\begin{equation}\label{matrix-P-bar}
\bar{\bm{P}}_{\bm{d}}\!=\!\DIAG\{\bpd\}(\bI-\DIAG\{\bpd\})+\bpd\bpd^T.
\end{equation}
%
%
%\begin{equation}\label{matrix-P-bar}
%\mbox{independent}~x_k\mbox{s}:\bar{\bm{P}}_{\bm{d}}\!=\!\DIAG\{\bpd\}(\bI-\DIAG\{\bpd\})+\bpd\bpd^T,~\mbox{fully correlated}~x_k\mbox{s}: \bar{\bm{P}}_{\bm{d}} = p_d \bone\bone^T.
%\end{equation}
%
Substituting (\ref{final-mean-case-A}), (\ref{final-var-case-A}) into  \eqref{DF-conditional} we have
\begin{align}\label{DF-coh-PAC}
%\boxed{
&\MDC (\ba) = \frac{\ba^T \bb \bb^T \ba}{\ba^T \bK \ba + c}
\end{align}
\begin{align}
&\mbox{where}~\nonumber\\
&\bb = |\bH|(\bpd -\bpf),~ c = M\frac{\sigma_{n}^2}{2},~ \bK = |\bH|(\bar{\bm{P}}_{\bm{d}} - \bpd\bpd^T)|\bH \nonumber
\end{align}
%
%=============coherent MAC==========================================
\subsection{\textbf{MAC}} Considering the signal model in (\ref{y-k}) and (\ref{fusion-statistic}), we have $\Rel(y) \!=\!  \sum_{k=1}^M  a_k |h_k|u_k + w$ where $ w \! \sim \! {\cal N}(0,\frac{\sigma_{n}^2}{2})$. We write $T \!= \!\ba^T |\bH| \bou + w$. Therefore
%given $\bou$ and $\bh$, we find
$\mE\{T |\bou, \bh\}\!=\! \ba^T |\bH| \bou$.
% and variance $\Var\{ T |\bou, \bh\}\!=\! \frac{\sigma_{n}^2}{2}$.
Substituting $\mE\{T |\bou, \bh\}$ into \eqref{mean-csi}  and applying similar facts as stated above, we find
\begin{equation}\label{final-mean-case-A-MAC}
\mE\{T|\cH_1, \bh\} \!=\! \ba^T |\bH| \bpd,~~~\mbox{and}~~~ \mE\{T|\cH_0, \bh\} \!= \! \ba^T |\bH| \bpf.
\end{equation}
Next, we find $\Delta$ and $\Delta'$. Since $T - \mE\{T|\bou, \bh\}=w$, we find $\Delta \!= \! \mE\{ w^2\} = \frac{\sigma_{n}^2}{2}$. Also, since $\mE\{T|\bou, \bh\} - \mE\{T|\cH_1, \bh\} \!=\! \ba^T |\bH| (\bou-\bpd), $ we have $\Delta' \!=\!  \ba^T |\bH|(\bou- \bpd)(\bou- \bpd)^T |\bH| \ba $. Substituting $\bar{\Delta} =\Delta+\Delta'$ into (\ref{var0-csi}) and using similar facts as stated above we reach
\begin{equation}\label{final-var-case-A-MAC}
\Var\{T|\cH_1,\bh\} = \frac{\sigma_n^2}{2} + \ba^T |\bH|( \bar{\bm{P}}_{\bm{d}}- \bpd\bpd^T)|\bH|\ba.
\end{equation}
Substituting (\ref{final-mean-case-A-MAC}), (\ref{final-var-case-A-MAC}) into  \eqref{DF-conditional} we have
\begin{align}\label{DF-coh-MAC}
%\boxed{
	&\MDC (\ba) = \frac{\ba^T \bb \bb^T \ba}{\ba^T \bK \ba + c}
\end{align}
\begin{align}
	&\mbox{where}~\nonumber\\
	&\bb = |\bH|(\bpd -\bpf),~ c =\frac{\sigma_{n}^2}{2},~ \bK = |\bH|(\bar{\bm{P}}_{\bm{d}} - \bpd\bpd^T)|\bH| \nonumber
\end{align}
Regarding the results in (\ref{DF-coh-PAC}) and (\ref{DF-coh-MAC}), a remark follows.

{\bf Remark}: For both PAC and MAC, the MDC takes the following form
\begin{equation}\label{summary}
  \MDC (\ba) = \frac{\ba^T \bb \bb^T \ba}{\ba^T \bK \ba + c}.
\end{equation}
Vector $\bb$ and matrix $\bK$ are identical for PAC and MAC, whereas scalar $c$, which captures the effect of the channel noises, is $M$ times larger in PAC. Note that $\bb$ and $\bK$ depend on the channel amplitudes $|\bH|$ and the local performance indices. Furthermore, $\bK$ depends on the spatial correlation among sensors' observations.
%========================================================
%\vspace{-0.4cm}
\section{MDC-Based Transmit Power Allocation}\label{power-alloc}
%\vspace{-0.2cm}
Recall  ${\cal P}_{k} \!= \! {\cal P}_{t_k} \theta_k$ where ${\cal P}_{t_k} $  is transmit power of sensor $k$ and $\theta_k$ captures the pathloss effect. Since $a_k\!=\!\sqrt{{\cal P}_{k}}$, we define $a_{t_k}\!=\!\sqrt{{\cal P}_{t_k}} \!= \! \frac{a_k}{\sqrt{\theta_k}}$.
Let $\ba_t\!=\![a_{t_1},...,a_{t_M}]^T$, $\bcP_t\!=\![\cP_{t_1},...,\cP_{t_M}]^T$, and $\sqrt{\bTh}$ be the component-wise square  root of
$\bTh \! = \!\DIAG\{[\theta_1,...,\theta_M]^T\}$.
We can rewrite (\ref{summary}) explicitly in terms of vector $\ba_t$ as
\begin{equation}\label{summary-with-theta}
\MDC(\ba_t) = \frac{\ba_t^T \bb_t \bb_t^T \ba_t}{\ba_t^T \bK_t \ba_t + c},
\end{equation}
where $\bb_t \!= \!\sqrt{\bTh} \bb$ and $\bK_t \!= \!\sqrt{\bTh} \bK \sqrt{\bTh}$. In this section, we maximize the MDC in (\ref{summary-with-theta}), with respect to $\ba_t$, subject to different sets of power constraints specified in Section \ref{system-model}: ($A$) TPC, where $ \ba_t^T \ba_t \! \leq \!{\cal P}_{tot}$; ($B$) TIPC, where  $ \ba_t^T \ba_t \! \leq \!{\cal P}_{tot}$ and $\bzero\preceq \ba_{t} \preceq \sqrt{{\bcP_0}}$. We define vector ${\bcP_0}=[{\cal P}_{0_1}, ..., {\cal P}_{0_M}]^T$ and $\sqrt{{\bcP_0}}$ is the component-wise square  root of ${\bcP_0}$; ($C$) IPC, where $\bzero\preceq \ba_{t} \preceq \sqrt{{\bcP_0}}$. Sections \ref{MDC-max-1},  \ref{MDC-max-2}, \ref{MDC-max-3} discuss the analytical solutions for MDC-based power allocations under these different sets of power constraints.
%
%==============================================================
\vspace{-0.4cm}
\subsection{Maximizing MDC in \eqref{summary-with-theta} under TPC}\label{MDC-max-1}
\vspace{-0.1cm}
The MDC-based transmit power allocation under TPC is the solution to the following problem
\begin{align*}
\begin{array}{cc}
  \begin{array}{cc}
  \max\limits_{\ba_t}. & ~~\frac{\ba_t^T \bb_t \bb_t^T \ba_t}{\ba_t^T \bK_t \ba_t + c}~~~~(\cO_1) \\
  \mbox{s.t.} &  ~~\ba_t^T\ba_t \leq {\cal P}_{tot} \\
   & ~~\ba_t\succeq \bzero
   \end{array}
 \end{array}
\end{align*}
We start with Lemma \ref{DF-ratio} which states that the solution to $(\cO_1)$ satisfies TPC at equality.
%\vspace{-0.3cm}
\begin{lem}\label{DF-ratio}
The maximum values of MDC in \eqref{summary-with-theta} are achieved when the inequality constraint $ \ba_t^T \ba_t \! \leq \!{\cal P}_{tot}$  turns into equality constraint.
\end{lem}
%
%\vspace{-0.3cm}
\begin{proof}
Suppose $\ba_{t1}$ maximizes MDC and $\ba_{t1}^T \ba_{t1} \!<\!{\cal P}_{tot}$. Define $\ba_{t2} \!=\! \frac{\ba_{t1} \sqrt{{\cal P}_{tot}}}{||\ba_{t1}||}$, which satisfies $\ba_{t2}^T \ba_{t2} \!=\! {\cal P}_{tot}$. We have
$\MDC(\ba_{t2}) = \frac{\ba_{t1}^T \bb_t \bb_t^T \ba_{t1}}{\ba_{t1}^T \bK_t \ba_{t1} + c(\frac{\ba_{t1}^T\ba_{t1}}{{\cal P}_{tot}} )}>\frac{\ba_{t1}^T \bb_t \bb_t^T \ba_{t1}}{\ba_{t1}^T \bK_t \ba_{t1} + c } =\\ \MDC(\ba_{t1}),$
which contradicts the optimality assumption of $\ba_{t1}$ i.e., the $\ba_{t}$ that maximizes MDC, must satisfy $\ba_{t}^T \ba_{t} \!=\! {\cal P}_{tot}$.
\end{proof}
%
%================================
%
When the inequality constraint in TPC is turned into equality constraint, we can rewrite MDC in \eqref{summary-with-theta} as
%$\MDC(\ba_t) \!=\! \frac{\ba_t^T \bb_t \bb_t^T \ba_t}{\ba_t^T \bK_t \ba_t +c \frac{\ba_t^T\ba_t}{{\cal P}_{tot}} } \!= \!\frac{\ba_t^T \bb_t \bb_t^T \ba_t}{\ba_t^T \bQ_a \ba_t}$,
%which results in
%
\begin{align}\label{equations-for-Q}
\MDC(\ba_t) \!= \! \frac{\ba_t^T \bb_t \bb_t^T \ba_t}{\ba_t^T \bQ_a \ba_t},~~
\mbox{where}~~ \bQ_a \!= \! \bK_t + \frac{c}{{\cal P}_{tot}}\bI.
\end{align}
%\begin{eqnarray}\label{equations-for-Q}
%\!\!\!\!\!\!\!\! \MDC(\ba_t) \!\! & \!\! =  \!\! & \!\! \frac{\ba_t^T \bb_t \bb_t^T \ba_t}{\ba_t^T \bQ_a \ba_t},~\mbox{where}~\bQ_a \!= \! \bK_t + \frac{c}{{\cal P}_{tot}}\bI, \nonumber \\
%\MDC(\bcP_t) \!\! & \!\! =  \!\! & \!\! \frac{\bcP_t^T \bb_t \bb_t^T \bcP_t}{\bcP_t^T \bQ_p \bcP_t},~\mbox{where}~\bQ_p \!= \! \bK_t + \frac{1}{2{\cal P}_{tot}} (\bd_t \bone^T + \bone \bd_t^T)+ \frac{c}{{{\cal P}_{tot}^2}}\bone\bone^T.
%\end{eqnarray}
%
Hence, $(\cO_1)$ reduces to
\begin{align*}
\begin{array}{cc}
  \begin{array}{cc}
  \max\limits_{\ba_t}. & ~~\frac{\ba_t^T \bb_t \bb_t^T \ba_t}{\ba_t^T \bQ_a \ba_t}~~~~(\cO'_1) \\
  \mbox{s.t.} &  ~~\ba_t^T\ba_t = {\cal P}_{tot} \\
   & ~~\ba_t\succeq \bzero
   \end{array}
 \end{array}
\end{align*}
To analytically solve $(\cO'_1)$, we use the result of Lemma \ref{max-dir} given below.
%\vspace{-0.5cm}
\begin{lem}\label{max-dir}
For $\bQ  \! \succ \! \bzero$ the function $f(\bx) \!= \! \frac{\bx^T \bb_t \bb_t^T \bx}{\bx^T \bQ \bx}$ is maximized at $\bx^* \!=\! \bQ^{-1} \bb_t$ and its non-zero scales.
\end{lem}
%\vspace{-0.5cm}
\begin{proof} See Appendix \ref{appendix-proof-lemma2}.
\end{proof}
%
%=======================
To be able to use Lemma \ref{max-dir} to solve $(\cO'_1)$, we need to examine whether symmetric matrix $\bQ_a$ is positive definite. Note $\bar{\bm{P}}_{\bm{d}} - \bpd\bpd^T \succ \bzero$ since it is a covariance matrix. Thus $\bK,\bK_t \succ \bzero$. Also $\frac{c}{{\cal P}_{tot}}\bI  \succ \bzero$. Therefore $\bQ_a \succ \bzero$.
To solve $(\cO'_1)$, we find $\hat{\bq} = \frac{\bq}{||\bq||}$ where $\bq = \bQ_a^{-1}\bb_t$. If $\hat{\bq} \succeq \bzero $ we let $\ba_{t}^* = \hat{\bq}\sqrt{{\cal P}_{tot}}$ and if $-\hat{\bq} \succeq \bzero$ we let $\ba_{t}^* = -\hat{\bq}\sqrt{{\cal P}_{tot}}$. But if all the entries of $\hat{\bq}$ do not have the same sign, we resort to numerical solutions. In particular, we turn the problem $(\cO'_1)$ into a convex problem and solve it numerically. We discuss these numerical solutions in Section \ref{convex-format}.
%=====================================

%\vspace{-0.1cm}
$\bullet$ {\it Analytical Solution to $(\cO'_1)$ with Independent Observations}: $\bar{\bm{P}}_{\bm{d}}$ is given in (\ref{matrix-P-bar}) and $\bK$ simplifies to $\bK \!=\! |\bH|\DIAG\{\bpd\} (\bI - \DIAG\{ \bpd\})|\bH|$. Let $g_k\!=\!\sqrt{\theta_k}|h_k|$. It is easy to verify $\bQ_a$ is a diagonal matrix with diagonal entries $[\bQ_a]_{kk}\!=\!  p_{d_k} ( 1- p_{d_k})g_k^2 + \frac{c}{{\cal P}_{tot}} $. Let $q_k$ be the $k$th entry of $\bq \!= \! \bQ_a^{-1} \bb_t$. Then
$q_k = \frac{ (p_{d_k} - p_{f_k}) g_k}{p_{d_k} ( 1 - p_{d_k}) g_k^2 + \frac{c}{{\cal P}_{tot}} },~~k=1,...,M$,
which is positive for $p_{d_k} \! > \! p_{f_k}$. We observe $q_k \! \approx \! \frac{(p_{d_k} - p_{f_k})}{p_{d_k} (1-p_{d_k}) g_k}$
for large $\frac{{\cal P}_{tot}}{c}$,  whereas $q_k \! \approx \! \frac{{\cal P}_{tot}}{c}(p_{d_k} - p_{f_k}) g_k$ for small $\frac{{\cal P}_{tot}}{c}$. For homogeneous sensors where
$p_{f_k}\!=\!p_f$ and $p_{d_k}\!=\!p_d$, we find the MDC-based power allocation strategy as $q_k \! \propto \! \frac{1}{g_k}$ for large $\frac{{\cal P}_{tot}}{c}$ (inverse water filling) and $q_k \! \propto \! g_k$ for small $\frac{{\cal P}_{tot}}{c}$ (water filling).
\vspace{-0.5cm}
\subsection{Maximizing MDC in \eqref{summary-with-theta} under TIPC}\label{MDC-max-2}
\vspace{-0.1cm}
The MDC-based transmit power allocation is the solution to the following problem
\begin{align*}
\begin{array}{cc}
 \begin{array}{cc}
 \max\limits_{\ba_t}. & ~~\frac{\ba_t^T \bb_t \bb_t^T \ba_t}{\ba_t^T \bK_t \ba_t+c}~~~~(\cO_2) \\
  \mbox{s.t.} &  ~~\ba_t^T\ba_t \leq {\cal P}_{tot} \\
   & ~~\bzero \preceq \ba_t \preceq \sqrt{\bcP_0}
    \end{array}
\end{array}
\end{align*}
While analytical solution to $(\cO_2)$ remains elusive, we find sub-optimal power allocation via solving the following optimization problem
\begin{align*}
\begin{array}{cc}
 \begin{array}{cc}
 \max\limits_{\ba_t}. & ~~\frac{\ba_t^T \bb_t \bb_t^T \ba_t}{\ba_t^T \bQ_a \ba_t}~~~~(\cO'_2) \\
  \mbox{s.t.} &  ~~\ba_t^T\ba_t = {\cal P}_{tot} \\
   & ~~\bzero \preceq \ba_t \preceq \sqrt{\bcP_0}
    \end{array}
\end{array}
\end{align*}
where $\bQ_a$ is given in (\ref{equations-for-Q}). Note that $(\cO'_2)$ is identical to $(\cO_2)$, except that  the inequality in TPC is turned into equality, i.e., the feasible set of $(\cO'_2)$ is a subset of the feasible set of $(\cO_2)$ and the objective function of $(\cO_2)$ is rewritten accordingly.
Indeed, this sub-optimal solution is an accurate solution when $\kappa\!=\!\frac{{\cal P}_{tot}\bg^T\bg}{c} \! \ll \! 1$ for $(\cO_2)$, as we show in the following.
Examining $\bK$ and $\bK_t$ when $\kappa \! \ll \! 1$, we can establish the following inequalities
\begin{align*}
&\ba_t^T\bK_t\ba_t \overset{(a)}{\leq} \ba_t^T\sqrt{\bTh}|\bH|\bone \bone^T|\bH|\sqrt{\bTh}\ba_t=\ba_t^T \bg \bg^T \ba_t \overset{(b)}{\leq}\\
&(\ba_t^T\ba_t) (\bg^T\bg) \overset{(c)}{\leq} {\cal P}_{tot} \bg^T\bg \overset{(d)}{\ll } c
\end{align*}
where $(a)$ is obtained noting that all entries of $\bar{\bm{P}}_{\bm{d}} - \bpd\bpd^T$ are less that $1$, $(b)$ is found using Cauchy-Schwarz inequality, $(c)$ comes from the inequality constraint in ($\cO_2$), and $(d)$ is due to $\kappa \! \ll \! 1$. This implies that when $\kappa \! \ll \! 1$, $(\cO_2)$ can be approximated as ($\cO_2^l$) in (\ref{O3-O4-low}).\\
%\vspace{-0.3cm}
%
\begin{align}\label{O3-O4-low}
 \begin{array}{cc}
 \min\limits_{\ba_t}. & ~~\frac{c}{\ba_t^T \bb_t \bb_t^T \ba_t}~~~~(\cO_2^l) \\
  \mbox{s.t.} &  ~~\ba_t^T\ba_t \leq {\cal P}_{tot} \\
   & ~~\bzero \preceq \ba_t \preceq \sqrt{\bcP_0}
    \end{array}
\end{align}
In Appendix \ref{appendix-proof-O3-TPC}, we show that the solution to $(\cO_2^l)$ satisfies the equality $\ba_t^T\ba_t \!=\! {\cal P}_{tot}$. This confirms that the solution to $(\cO'_2)$ (sub-optimal solution) is an accurate substitute for the solution to $(\cO_2)$ under the condition $\kappa \! \ll \! 1$.
%We proceed to solve $(\cO_2')$, $(\cO_4')$. %To solve these problems, we use the result of Lemma 3 given below, which is built on Lemma 2.
%%\vspace{-0.6cm}
%\begin{lem}\label{single-max}
%Lemma 2 states that $f(\bx) \!= \! \frac{\bx^T \bb_t \bb_t^T \bx}{\bx^T \bQ \bx}$ has a maximum at $\hat{\bx}=\frac{\bx^*}{||\bx^*||}$ when $\bQ \succ \bzero$. If $\hat{\bx} \succeq \bzero$ then $f(\bx)$ does not have any other local maximum or minimum in the set $\{\bx: \bx \succeq \bzero\}$.
%\end{lem}
%%
%%\vspace{-0.6cm}
%\begin{proof}
%See Appendix \ref{appendix-proof-lemma3}.
%\end{proof}
%%
To solve $(\cO'_2)$, we first ignore the box constraints of IPC and consider only TPC at equality. The problem solving strategy is similar to solving $(\cO'_1)$ in Section \ref{MDC-max-1}. In particular, to solve $(\cO_2')$, we find $\hat{\bq} \!= \!\frac{\bq}{||\bq||}$ where $\bq \! = \! \bQ_a^{-1}\bb_t$. If $\hat{\bq} \! \succeq \! \bzero $ we let $\ba_{t1}^* \! = \! \hat{\bq}\sqrt{{\cal P}_{tot}}$ and if $-\hat{\bq} \! \succeq \! \bzero$ we let $\ba_{t1}^* \! = \! -\hat{\bq}\sqrt{{\cal P}_{tot}}$.
If $\ba_{t1}^* $ satisfies the box constraint $\bzero \!  \preceq \! \ba_t \! \preceq \! \sqrt{\bcP_0}$, it is the solution to $(\cO'_2)$.
However, if $\ba_{t1}^* $ does not satisfy its corresponding box constraint, following Appendix \ref{appendix-proof-lemma2}, we can easily show that $f(\bx) \!= \! \frac{\bx^T \bb_t \bb_t^T \bx}{\bx^T \bQ \bx}$ does not have local maximum or minimum in the set $\{\bx: \bx \succeq \bzero\}$. This means that, in this case, the closest feasible point to $\ba_{t1}^*$ is the solution to $(\cO'_2)$. That is, the solution to  $(\cO'_2)$ when $ \bzero \! \preceq \! \ba_{t1}^* \! \npreceq \! \sqrt{\bcP_0} $ is the solution to $(\cO''_2)$ given below
\begin{align*}
\begin{array}{cc}
  \begin{array}{cc}
 \min\limits_{\ba_t}. & ~~|\ba_t -\ba_{t1}^*|^2 ~~~~(\cO''_2) \\
  \mbox{s.t.} &  ~~\ba_t^T\ba_t = {\cal P}_{tot} \\
   & ~~\bzero \preceq \ba_t \preceq \sqrt{\bcP_0}
   \end{array}
 \end{array}
\end{align*}
%
%The argument follows. In Appendix \ref{appendix-proof-lemma2}, we showed that $f(\bx) \!= \! \frac{\bx^T \bb_t \bb_t^T \bx}{\bx^T \bQ \bx} \! \leq \! \lambda_{max}(\bar{\bb_t}\bar{\bb_t}^T)$, where  $\bar{\bb_t}\bar{\bb_t}^T$ is a rank one matrix with the eigenvalue $|\bar{\bb_t}|^2$ and the eigenvector $\bar{\bb_t}$. Thus, $f(\bx)$ achieves its maximum at $\bx^* \!= \! \bQ^{-1}\bb_t $ and does not have local maximum or minimum in the set $\{\bx: \bx \succeq \bzero\}$.
%
Our analytical solution to $(\cO''_2)$ is presented in the appendix \ref{coh-case2}. Note that $(\cO''_2)$ is not convex. In Appendix \ref{coh-case2} we show that, despite this fact, the solution to Karush-Kuhn-Tucker (KKT) conditions for $(\cO''_2)$ is unique.
%==========================================================
%\vspace{-0.2cm}
\subsection{Maximizing MDC in \eqref{summary-with-theta} under IPC}\label{MDC-max-3}
%\vspace{-0.3cm}
The MDC-based transmit power allocation is the solution to the following optimization problem
\begin{align*}
  \begin{array}{cc}
  \max\limits_{\ba_t}. & ~~\frac{\ba_t^T \bb_t \bb_t^T \ba_t}{\ba_t^T \bK_t \ba_t + c}~~~~(\cO_3) \\
  \mbox{s.t.}  & ~~\bzero \preceq \ba_t \preceq \sqrt{\bcP_0}
  \end{array}
\end{align*}
Similar to Section \ref{MDC-max-2}, we show below that, when $\xi\!=\!\frac{\bone^T \bcP_0 \bg^T\bg}{c} \! \ll \! 1$, $(\cO_3)$ can be approximated as $(\cO_3^l)$ in (\ref{O5-O6-low}).
Examining $\bK$ and $\bK_t$ when $\xi \! \ll \! 1$, we can establish the following inequalities
\begin{align*}
&\ba_t^T\bK_t\ba_t \overset{(a)}{\leq} \ba_t^T\sqrt{\bTh}|\bH|\bone \bone^T|\bH|\sqrt{\bTh}\ba_t = \ba_t^T\bg \bg^T\ba_t
\overset{(b)}{\leq} \\
&(\ba_t^T\ba)(\bg^T\bg)  \overset{(c)}{\leq}\bone^T \bcP_0 \bg^T\bg  \overset{(d)}{\ll } c,
\end{align*}
where $(a)$ is because all entries of $\bar{\bm{P}}_{\bm{d}} - \bpd\bpd^T$ are less that $1$, $(b)$ is found using Cauchy-Schwarz inequality, $(c)$ comes from the inequality in IPC, and $(d)$ is due to $\xi \! \ll \! 1$. This implies that when $\xi \! \ll \! 1$, $(\cO_3)$ can be approximated with ($\cO_3^l$) in (\ref{O5-O6-low}).
\begin{align}\label{O5-O6-low}
 \begin{array}{cc}
 \min\limits_{\ba_t}. & ~~\frac{c}{\ba_t^T \bb_t \bb_t^T \ba_t}~~~~(\cO_3^l) \\
  \mbox{s.t.} &  ~~\bzero \preceq \ba_t \preceq \sqrt{\bcP_0}
    \end{array}
\end{align}
In Appendix \ref{appendix-proof-O5-IPC} we show that the solution to $(\cO_3^l)$ is $\ba_t=\sqrt{\bcP_0}$.
%Below, we provide the analytical solution to $(\cO_3)$ for the special case of independent observations.

%=====================================
%\vspace{-0.2cm}
$\bullet$ {\it Analytical Solution to $(\cO_3)$  with Independent Observations}: Suppose $\cP_{0_k} \! = \! \cP_0,~k\! =\!1,...,M$. We showed in Section \ref{MDC-max-1} that $\bK,\bK_t \succ \bzero$. With independent observations, these matrices become diagonal.
To solve $(\cO_3)$, we minimize $\frac{1}{\MDC(\ba_t)}$ under IPC. Assume $\bpsi$, $\bphi$ are the Lagrange multipliers of the constraints $\ba_t \! \preceq \! \sqrt{\bcP_0}$ and $ \ba_t  \! \succeq \!\bzero $, respectively. Then KKT conditions are
\begin{align}\label{eta-definition}
& \frac{2}{(\bb_t^T \ba_t)^2} ([\bK_t]_{kk} a_{t_k} - b_{t_k} \eta)  +  \psi_k - \phi_k \!=\! 0, ~~~k=1,...,M,\\
&\mbox{where}~ \eta \!=\! \frac{\ba_t^T \bK_t \ba_t + c}{\bb_t^T \ba_t} \nonumber\\
&\psi_k (a_{t_k}-\sqrt{\cP_{0}})\!=\!0,~~\psi_k \geq 0,~~a_{t_k} \leq \sqrt{\cP_{0}},\nonumber\\
&\mbox{and}~~
\phi_k a_{t_k}=0,~~\psi_k\geq 0,~~ a_{t_k} \geq 0, \nonumber
\end{align}
Since $\ba_t^T \bK_t \ba_t \! \geq \! 0$, $\bb_t \! \succ \! 0$ and $c\!  > \! 0$, we find $\eta \! > \! 0$. Solving the KKT conditions yields %$a_{t_k} \neq 0$ and
\begin{equation}\label{a-tk}
a_{t_k} = \left\{
\begin{array}{ll}
\frac{b_{t_k} }{ [\bK_{t}]_{kk}  } \eta, & \quad\text{for } \eta \leq \frac{[\bK_{t}]_{kk} }{b_{t_k}}\sqrt{\cP_{0}}\\
\sqrt{\cP_{0}}, & \quad\text{otherwise}.
\end{array}\right.
\end{equation}
In Appendix \ref{uniqueness}, we show that at least one of $a_{t_k}$s in \eqref{a-tk} obtains its maximum $\sqrt{\cP_{0}}$.
Suppose we sort the sensors such that
$\frac{b_{t_{i_1}}}{[\bK_t]_{i_1 i_1}} \! \geq \! ... \! \geq \!\frac{b_{t_{i_M}}}{[\bK_t]_{i_M i_M}}$, i.e., $a_{t_{i_1}} \! \geq \! ... \! \geq \! a_{t_{i_M}}$. Let ${\ba}_t\!=\![a_{t_{i_1}},..., a_{t_{i_m}}, a_{t_{i_{m+1}}},..., a_{t_{i_M}}]^T$ and $a_{t_{i_1}} \!= \! ... \!= \! a_{t_{i_m}} \!= \! \sqrt{\cP_{0}},~1 \leq m \leq M$. Solving ${\ba}_t^T \bK_t {\ba}_t + c - \eta \bb_t^T {\ba}_t \!= \! 0$ for $\eta$, combined with \eqref{a-tk}, we find
%
%\begin{equation}
$\eta_0 = \frac{\cP_{0} \sum_{j=1}^{m} [\bK_t]_{i_j i_j}  + c }{\sqrt{\cP_{0}} \sum_{j=1}^{m}b_{t_{i_j}}}$.
%\end{equation}
%
If $\frac{ [\bK_t]_{i_m i_m}}{b_{t_{i_m}}}\sqrt{\cP_{0}} \leq \eta_0 \leq \frac{[\bK_t]_{i_{m+1} i_{m+1}}}{b_{t_{i_{m + 1 }}}}\sqrt{\cP_{0}}$, then the above assumption is valid, and we substitute $\eta_0$ in \eqref{a-tk} and calculate $a_{t_{i_{m+1}}},..., a_{t_{i_M}}$ and $\MDC$ in $(\cO_3)$. Note that $\eta_0$ depends on $m$.
Otherwise, we increase $m$ by one and repeat the procedure, until we reach $\eta_0$ that lies within the proper interval. In Appendix \ref{uniqueness} we also show that, although $(\cO_3)$ is not convex, the KKT solution in (\ref{a-tk}) is unique.

%==================================================
%\vspace{-0.4cm}
\subsection{Discussion on Maximization of MDC Using Convex Optimization Program}\label{convex-format}
%\vspace{-0.2cm}
Recall that in Section \ref{MDC-max-1} we could not find a closed form solution for $(\cO_1)$ when all the entries of $\hat{\bq}$ do not have the same sign. Also, the analytical solution to $(\cO_2)$, formulated in Section \ref{MDC-max-2}, remains elusive. Hence, we have provided a sub-optimal solution, via solving $(\cO_2')$ that is accurate solution when $\kappa \! \ll \! 1$. Similarly, we have derived a sub-optimal solution to $(\cO_3)$, formulated in Section \ref{MDC-max-3}, that is accurate solution when $\xi \! \ll \! 1$. In this section, we turn $(\cO_1)$,$(\cO_2)$,$(\cO_3)$ into convex optimization problems, in order to solve them numerically using CVX program.

We start with $(\cO_2)$, in which we wish to minimize $\frac{1}{\MDC(\ba_t)} \!= \! \frac{  \ba_t^T \bK_t \ba_t+c}{(\bb_t^T \ba_t)^2}$, under TIPC. Let $\bx_a \!=\! \frac{\ba_t}{\bb_t^T\ba_t}$ and $t_a \!= \! \frac{1}{\bb_t^T\ba_t}$. Therefore $\bb_t^T \bx_a\!=\!1$ and $\ba_{t}\!=\!\frac{\bx_a}{t_a}$. Employing these definitions, $(\cO_2)$ can be rewritten in the following equivalent form
\begin{align}\label{O31-O41}
\begin{array}{cc}
  \begin{array}{cc}
  \min\limits_{\bx_a,t_a}. & ~~ \bx_a^T \bK_t \bx_a + ct_a^2 ~~~~(\cO_{2_1}) \\
  \mbox{s.t.}  &~~\bx_a^T\bx_a \leq {\cal P}_{tot}t_a^2\\
  & ~~\bx_a \preceq t_a \sqrt{\bcP_0}\\
& ~~\bzero \preceq \bx_a \\
  &~~\bb_t^T \bx_a=1
  \end{array}
\end{array}
\end{align}
We can reformulate $(\cO_{2_1})$ as
\begin{align*}\label{O32-O42}
\begin{array}{cc}
\begin{array}{cc}
  \min\limits_{\bz_a}. & ~~ \bz_a^T \bD_a \bz_a  ~~~~(\cO_{2_2}) \\
  \mbox{s.t.}  &~~\bz_a^T \left[
                            \begin{array}{cc}
                               \bI & \bzero \\
                               \bzero^T & -{\cal P}_{tot}\\
                            \end{array}
                            \right]  \bz_a <1\\
  & ~~[\bI , -\sqrt{\bcP_0}] \bz_a  \preceq \bzero\\
&~~[\bb_t^T , 0] \bz_a =1,~~\bz_a \succeq  \bzero
  \end{array}
\end{array}
\end{align*}
where
\begin{equation*}
\bz_a = [\bx_a^T , t_a]^T,
~~\bD_a= \left[
                   \begin{array}{cc}
                    \bK_t & \bzero\\
                    \bzero^T & c\\
                    \end{array}
             \right].
\end{equation*}
Examining $(\cO_{2_2})$, we realize that it is a quadratic programming (QP) convex problem since $\bD_a \! \succeq \! \bzero$ and hence it can be solved using CVX program.
One can take similar steps to turn $(\cO_1)$ and $(\cO_3)$ into a problem whose optimal solution can be found using CVX program. In particular, we formulate $(\cO_{1_1})$,$(\cO_{1_2})$ via deleting the second inequality constraints corresponding to IPC and $(\cO_{3_1})$,$(\cO_{3_2})$ by removing the first inequality constraints corresponding to TPC from $(\cO_{2_1})$,$(\cO_{2_2})$, respectively. Since $\bD_a \succeq \bzero$, $(\cO_{1_2})$ and $(\cO_{3_2})$ are also QP convex problems.
\vspace{-0.3cm}
\section{Numerical Results}\label{num-results}
\vspace{-0.3cm}
In this section, through simulations, we corroborate our analytical results. We study the effect of correlation between sensors' observations on the MDC, the performance improvements achieved by the MDC-based transmit power allocations (we refer to as ``DPA''), and the impact of different sensing and communication channels on DPA. For our simulations, we consider the signal model
${\cal H}_0: x_k\!=\!z_k,~{\cal H}_1: x_k\!=\!s_k+z_k, \mbox{for} ~k=1,...,M,$
where $z_k\sim {\cal N}(0,\sigma^2_0)$ and $s_k \sim {\cal N}(0,\sigma^2_{s_k})$ is a sample of an external Gaussian signal source $s \sim {\cal N}(0,\sigma^2_{s})$. We assume $\sigma^2_{s_k}\!=\!\frac{\sigma^2_s}{d^{\epsilon_s}_{S_k}}$, where $d_{S_k}$ is the distance between sensor $k$ and $s$ and $\epsilon_s$ is the pathloss exponent. We assume $z_k$ and $s_k$ are mutually uncorrelated, however, $s_k$'s are correlated. Let $\bs\!=\![s_1,s_2,...,s_M]^T$ have covariance matrix $K_s\!=\!\mathbb{E}\{\bs\bs^T\}$. We assume $[K_s]_{ij}\!=\!\rho_{ij}\sqrt{\sigma^2_{s_i}\sigma^2_{s_j}}$ where $\rho_{ij}\!=\!\rho^{d_{ij}}$, $0 \!\leq\! \rho\!\leq\!1$ is the correlation at unit distance and depends on the environment and $d_{ij}$ is the distance between sensors $i$ and  $j$ \cite{Hafeez}. Each sensor employs an energy detector that  maximizes $p_{d_k}$, under the constraint $p_{f_k} \!< \!0.1$. Sensors are deployed at equal distances from each other, on the circumference of a circle with diameter $5$m on the $x$-$y$ plane, where the coordinate of its center is $(0,0,0)$.
For sensing part, we assume $M\!=\!8$, $\epsilon_s\!=\!2$, $\sigma^2_s \!=\!5\,\mbox{dBm}$, $\sigma^2_0  \!=\!-70\,\mbox{dBm}$, and for communication part we let $\sigma^2_n \!=\!-70\,\mbox{dBm}$, $G\!= \!-55\,\mbox{dB}$ \cite{Poor_PA},  $\epsilon_c \!=\!2$, and $\cP_{0_1}\!=\! ... \!=\!\cP_{0_M}\!=\!\bar{\cP}$.
%{\red We consider an energy detector at each sensor and maximize $p_{d_k}$ at each sensor under the constraint $p_{f_k} < 0.1$}.
%=============================================================

\underline{\it Performance of DPA when $p_{d_k}$'s and pathloss are identical}:
Suppose the coordinates of signal source $s$ and the FC, respectively, are $(0,0,3\mbox{m})$, and $(0,0,-10\mbox{m})$.
With this configuration, $p_{d_k}\!=\!0.6615, \forall k$ and pathloss are identical. We assume $h_{k}\sim {\cal C}{\cal N}(0,1), \forall k$ and we average over 10,000 number of channel realizations to obtain the results. We explore the MDC enhancements achieved by DPA and compare the MDC values with those of obtained by uniform power allocation (we refer to as ``UPA''), in which sensors transmit at equal powers.

%=======================Fig 0 Optimal vs. MDC under TPC======================================
Fig. \ref{fig:LRT_Linear} compares optimal power allocation (OPA), which finds the sensors' powers that maximize $P_{D_0}$ under the constraint $P_{F_0} < \beta_F$, DPA and UPA, for linear fusion rule and the optimal LRT rule. To find OPA with both linear and the LRT rules and DPA with the LRT rule, we use brute force search to find the power values, and we simplify the network and only consider $s_1$ and $s_5$. We assume $\rho=0.1$, $\beta_F=0.1$  and plot $P_{D_0}$ versus ${\cal P}_{tot}$ under TPC for PAC. Fig. \ref{fig:Linear_2sensor} compares OPA, DPA and UPA, given linear fusion rule. We observe that at low ${\cal P}_{tot}$, they are close to each other but as ${\cal P}_{tot}$ increases, they diverge and DPA outperforms UPA but performs worse than OPA. Fig. \ref{fig:LRT_2sensor} compares OPA, DPA and UPA, given the LRT rule. Similarly, we see that DPA performs between UPA and OPA. We note that at low ${\cal P}_{tot}$, DPA approaches OPA. Also we plot $P_{D_0}$ for the LRT rule, where the transmit power values are obtained from maximizing the MDC of the linear fusion rule (in Fig. 1(b) we refer to it as "Power alloc. by DPA with linear rule"). We observe that, except at low ${\cal P}_{tot}$, this curve is very close to $P_{D_0}$ corresponding to DPA for the LRT rule, implying  that the MDC-based power allocation with linear fusion rule is very close to the MDC-based power allocation with LRT rule.

%=======================Fig 1 coherent under TPC======================================
%
Figs. \ref{fig:PD0-TPC-coh} and \ref{fig:MDC-TPC-coh} show $P_{D_0}$ and maximized MDC versus ${\cal P}_{tot}$, respectively, under TPC for PAC and MAC, $\rho\!=\!0.1, 0.9$, and $\beta_F=0.05$ in Fig. \ref{fig:PD0-TPC-coh}. Comparing Figs. \ref{fig:PD0-TPC-coh} and \ref{fig:MDC-TPC-coh}, we observe that the MDC and $P_{D_0}$ follow similar trends. Hence, to make our computations faster and less complex, in the rest of this section we only calculate the MDC. From Fig. \ref{fig:MDC-TPC-coh}, we note that the MDC increases by increasing ${\cal P}_{tot}$ or by decreasing $\rho$.
Comparing PAC and MAC, we note that MAC outperforms PAC at low ${\cal P}_{tot}$, whereas  PAC converges to MAC at high ${\cal P}_{tot}$. These are due to the facts that, at low ${\cal P}_{tot}$ the effect of communication channel noise characterized by $c$ in the MDC expression of PAC is $M$ times larger than that of MAC (see equations
(\ref{DF-coh-PAC}) and (\ref{DF-coh-MAC})) and thus MAC outperforms PAC. However, at high ${\cal P}_{tot}$ this difference in $c$ values is negligible and hence PAC converges to MAC. We also observe that, at low ${\cal P}_{tot}$ the performance gaps corresponding to DPA and UPA are negligible, despite the fact that sensors experience different communication channel fading.
This is because at low ${\cal P}_{tot}$ the dominant effect of communication channel noise renders the decisions of sensors equally important to the FC, regardless of the channel realizations and the actual (different) decisions.
On the other hand, at high ${\cal P}_{tot}$ the performance gaps corresponding to DPA and UPA are significant. Note that this performance gap in MAC is wider than that of PAC. This is expected, since the larger $c$ value in PAC undermines the differences between sensors and narrows the performance gap between DPA and UPA.
As $\rho$ increases, the chances that sensors make similar decisions increase and therefore the performance gaps between DPA and UPA shrink.

%=============Fig 2======================================================
%This behavior is proved analytically in appendix \ref{TPC-low_power}.
%With the same setup, we plot detection probability at the FC,$P_{D_0}$, versus ${\cal P}_{tot}$ under the constraint that the false alarm probability at the FC, $P_{F_0}$, is less than --- . We observe that $P_{D_0}$ behaves similar to MDC, but the probability values could change depending on the $P_{F_0}$ constraint.
% In case $2$, as case $1$, as we have shown before $\ba_t^T\ba_t={\cal P}_{tot}$ so $\ba^*_t=\bb_t\sqrt {\frac{{\cal P}_{tot}}{\bb_t^T\bb_t}}$ and since ${\cal P}_{tot}$ is small, it is very likely that it satisfies box constraint so $\ba_t=\bb_t\sqrt {\frac{{\cal P}_{tot}}{\bb_t^T\bb_t}}$. So again by substituting this optimal solution we obtain equal MDCs for optimal and uniform power allocations.
%========================================================================
%
%=================Fig 2 coherent under TIPC==================================
Fig. \ref{fig:MDC-TIPC-coh} shows the MDC maximized under TIPC versus ${\cal P}_{tot}$ for PAC and MAC, $\bar{\cP}\!=\!30\mW$, and $\rho\!=\!0.1, 0.9$.
%In addition to the setting in Fig. \ref{fig:MDC-TPC-coh}, the sensors powers are limited to $\bar{\cP}=30\mW$.
Similar to Fig. \ref{fig:MDC-TPC-coh}, the MDC increases by increasing ${\cal P}_{tot}$ or decreasing $\rho$ and MAC outperforms PAC. We also compare the MDC obtained from solving $(\cO_2)$ and $(\cO'_2)$, in which we have the inequality constraint (I) $|\ba_t|^2\!\leq\!{\cal P}_{tot}$ and the equality constraint (E) $|\ba_t|^2\!=\!{\cal P}_{tot}$, respectively.
We observe that at low ${\cal P}_{tot}$, there is no performance gap corresponding to ``DPA with E'' and ``DPA with I'', whereas at high ${\cal P}_{tot}$, the performance of ``DPA with E'' degrades from that of ``DPA with I''. This performance degradation in MAC is due to the increasing interference of sensors' decisions at the FC when sensors are assigned higher transmit power. At very high ${\cal P}_{tot}$ the performance of ``DPA with E'' reduces to that of UPA. This is because the maximum value that ${\cal P}_{tot}$ can assume is $M \bar{\cP}$. Hence, at very high ${\cal P}_{tot}$ we have $|\ba_t|^2\!=\!M \bar{\cP}$, implying that $a_{t_k}=\sqrt{\bar{\cP}}, \forall k$.
%
%=================Fig 3 coherent under IPC==================================
%
Fig. \ref{fig:MDC-IPC-coh} shows the MDC maximized under IPC versus $\bar{\cP}$ for PAC and MAC and $\rho\!=\!0.1, 0.9$. We note that at low $\bar{\cP}$, the performances of DPA and UPA are similar, since $a_{t_k}=\sqrt{\bar{\cP}}, \forall k$. This observation is in agreement with our analytical results in Section \ref{MDC-max-3}, where we showed for $\xi \! \ll \! 1$ we have $\ba_t=\sqrt{\bcP_0}$.
%
%=======================================================================
%In the above figures, we observe that as $\rho$ increases, the difference between DPA and UPA decreases. This is because by increasing $\rho$, information gathered in sensors are more correlated and sensors become more similar to each other from the FC view point. Figures also show that DPA and UPA performances are closer to each other in PAC than in MAC. Since, from (\ref{DF-coh-PAC}) and (\ref{DF-coh-MAC}), $c=M\frac{\sigma_n^2}{2}$ and $c=\frac{\sigma_n^2}{2}$ for PAC and MAC, respectively, receiver noise affects PAC more than MAC and reduces its final received SNR. This undermines the difference between the sensors in PAC which results in UPA as an optimal solution.
%=======================================================================
%
%===================Fig 4 noncoherent======================================
%
Examining the effect of increasing $\rho$ from $0.1$ to $0.9$ in Figs. \ref{fig:MDC-TPC-coh}, \ref{fig:MDC-TIPC-coh} and \ref{fig:MDC-IPC-coh}, we observe that the performance gap (i.e., the difference between the two maximized MDC values) increases as ${\cal P}_{tot}$ or $\bar{\cP}$ increases.

%===========================================================================
%
\underline{\it Trends of DPA  when $p_{d_k}$'s are different and pathloss are }
\\\underline{\it identical}:
We change the coordinate of signal source $s$ to $(2.5\mbox{m}, 0, 3\mbox{m})$, right above sensor $S_1$. With this configuration, $p_{d_k}$'s change to
$\bpd^T\!=\![0.7329, 0.6882, 0.6083, 0.5505, 0.5307, 0.5505, 0.6083, 0.6882]$, while pathloss are still identical (note that $p_{d_1}$, $p_{d_2}$, $p_{d_8}$ are the three largest).
Assuming $h_k\!=\!1, \forall k$, we investigate the impact of different $p_{d_k}$'s on DPA via plotting $\cP_{t_k}, \forall k$.
Consider Fig. \ref{fig:coh_MAC} which plots $\cP_{t_k}$ for MAC. Figs. \ref{fig:coh_MAC_case1_01}, \ref{fig:coh_MAC_case1_09} correspond to the case when the MDC is maximized under TPC, $\rho\!=\!0.1,0.9$, and ${\cal P}_{tot}\!=\!30\, \mW, \,120\, \mW, \, 240 \,\mW$. Figs. \ref{fig:coh_MAC_case2_01}, \ref{fig:coh_MAC_case2_09}  correspond to the case when the MDC is maximized under TIPC, $\rho\!=\!0.1,0.9$, $\bar{\cP}\!=\!30\, \mW $, and ${\cal P}_{tot}\!=\!30\, \mW, \,120\, \mW, \, 240 \,\mW$. Figs. \ref{fig:coh_MAC_case3_01}, \ref{fig:coh_MAC_case3_09}  correspond to the case when the MDC is maximized under IPC, $\rho\!=\!0.1,0.9$, and $\bar{\cP}\!=\!4\, \mW, \,15\, \mW, \, 30\, \mW$.
These figures show that for the cases when the MDC is maximized under TPC or under TIPC, sensors with higher $p_{d_k}$ values (i.e., more reliable local decisions) are assigned higher $\cP_{t_k}$ for all ${\cal P}_{tot}$ values.
For the case when the MDC is maximized under IPC, at low $\bar{\cP}$, $\cP_{t_k}=\bar{\cP}, \forall k$ (we have UPA). However, as $\bar{\cP}$ increases, for those sensors with smaller $p_{d_k}$ values (i.e., less reliable local decisions) we have smaller $\cP_{t_k}$.%, such that for $p_{d_i}\!< \!p_{d_j}$ we have $\cP_{t_i} \!<\!\cP_{t_j} \!\leq \! \bar{\cP}$.
We also note that as $\rho$ increases from 0.1 to 0.9 the variations of $\cP_{t_k}$ across sensors increase: for the case when the MDC is maximized under TPC, sensors with larger and smaller $p_{d_k}$'s, respectively, are assigned further more and lesser $\cP_{t_k}$; for the case when the MDC is maximized under TIPC or IPC, sensors with  smaller $p_{d_k}$'s are assigned less $\cP_{t_k}$, such that for $p_{d_i}\!< \!p_{d_j}$ we have $\cP_{t_i} \!<\!\cP_{t_j} \!\leq \! \bar{\cP}$.
Fig. \ref{fig:coh_PAC} plots $\cP_{t_k}$ for PAC. Comparing Figs. \ref{fig:coh_MAC} and \ref{fig:coh_PAC}, we note that similar trends hold true, while the variations of $\cP_{t_k}$'s across sensors in MAC, especially in TIPC and IPC, are wider than those of PAC (i.e., $\cP_{t_k}$'s across sensors in MAC are more different than UPA), due to the fact that the $c$ value in MAC is smaller.

%=======================================================
%
\underline{\it Trends of DPA when $p_{d_k}$'s are identical and pathloss are}\\
\underline{\it different}:
Suppose the coordinates of $s$ and the FC, respectively, are $(0,0,3\mbox{m})$, and $(2.5\mbox{m},0,-3\mbox{m})$, where the FC is right below sensor $S_1$. With this configuration, $p_{d_k}\!=\!0.6615, \forall k$, whereas the pathloss are different (note that the pathloss corresponding to $S_1$, $S_2$ and $S_8$ are the three smallest).
We observed that $\cP_{t_k}$'s for different $\rho$ values remain the same. Hence, in this part we focus on $\rho\!=\!0.1$. DPA is shown in figures \ref{fig:coh_MAC-pathloss2} and \ref{fig:coh_PAC-pathloss2} for MAC and PAC, respectively. We observe that, for both PAC and MAC under TPC or TIPC, sensors with larger pathloss are assigned higher $\cP_{t_k}$ (we refer to as inverse water filling). Examining the case when the MDC is maximized under IPC and $\bar{\cP}\!=\!4 \mW$, we have  $\cP_{t_k}=\bar{\cP}, \forall k$ (UPA) in PAC, whereas sensors with larger pathloss are assigned higher $\cP_{t_k}$ (inverse water filling) in MAC. This is due to the fact that the $c$ value in MAC is smaller and therefore, the effective received signal-to-noise ratio in MAC is larger, leading to variations of $\cP_{t_k}$'s across sensors.
To investigate more the effect of different pathloss on DPA, we move the FC further from the sensors and change its coordinate to $(2.5\mbox{m}, 0, -10\mbox{m})$, to effectively increase the pathloss between all the sensors and the FC (and decrease received power at the FC), while still $S_1$, $S_2$ and $S_8$ have the three smallest pathloss. We observe that in TPC and TIPC sensors with smaller pathloss are assigned higher $\cP_{t_k}$ (we refer to as water filling), whereas in IPC, $\cP_{t_k}=\bar{\cP}, \forall k$ (we have UPA).
\vspace{-0.4cm}
\section{Conclusion}\label{conc}
\vspace{-0.3cm}
We considered a channel aware binary distributed detection problem in a WSN with coherent reception and linear fusion rule at the FC, where observations are correlated Gaussian and sensors are unaware of such correlation when making decisions. Assuming that the sensors and the FC are connected via PAC or MAC, we studied power allocation schemes that maximize the MDC at the FC. Our numerical results suggest that when MDC-based power allocation and optimal transmit power allocation are employed at low $P_{tot}$, the resulting $P_{D_0}$ is very close for both linear fusion rule and the LRT rule. For homogeneous sensors with identical pathloss, MAC outperforms PAC at low ${\cal P}_{tot}$ under TPC and TIPC (low $\bar{\cP}$ under IPC), whereas PAC converges to MAC at high ${\cal P}_{tot}$. Compared with equal power allocation, performance enhancement offered by the MDC-based power allocation is more significant in MAC and this improvement reduces as correlation increases. For inhomogeneous sensors with identical pathloss, sensors with more reliable decisions are assigned higher powers. As correlation increases, the variations of power across sensors increase: sensors with more (less) reliable decisions, are assigned further more (lesser) powers. For homogeneous sensors with different pathloss, power allocations are invariant as correlation changes. At low (high) received power at the FC, sensors with smaller (larger) pathloss are assigned higher powers under TPC and TIPC.
%
%=====================================================================
\vspace{-0.5cm}
%\section{Appendix}
\appendix
%%\vspace{-0.3cm}
%===============Appendix A===============================================
%
\vspace{-0.3cm}
\subsection{Proof of Lemma \ref{max-dir}}\label{appendix-proof-lemma2}
%%\vspace{-0.3cm}
%
%Clearly, if $\bx^*$ maximizes $f(\bx)$, any non-zero scale of $\bx^*$ also maximizes $f(\bx)$.
Consider $\bQ \! =\! \bD \bLa \bD^T$, where $\bLa\!=\!\DIAG\{[\lambda_1,...,\lambda_M]^T\}$ and $\lambda_k$'s are the positive eigenvalues of $\bQ$ and columns of $\bD$ are the eigenvectors of $\bQ$. We can rewrite $f(\bx)$ as
%
%\begin{equation*}
$f(\bx) \!= \! \frac{(\bx^T (\bD\sqrt{\bLa})(\bD\sqrt{\bLa})^{-1}\bb_t )^2}{\bx^T\bD\sqrt{\bLa}\sqrt{\bLa}\bD^T \bx} \!= \!\frac{\bar{\bx}^T \bar{\bb_t} \bar{\bb_t}^T \bar{\bx}}{\bar{\bx}^T\bar{\bx}}$,
%\end{equation*}
%
where $\bar{\bx} \!= \! \sqrt{\bLa}\bD^T\bx$ and $\bar{\bb_t} \!= \!(\bD\sqrt{\bLa})^{-1} \bb_t$. Using the Rayleigh Ritz inequality \cite{matrix-diff-Book}, we find $f(\bx)  \leq \lambda_{max}(\bar{\bb_t}\bar{\bb_t}^T)$ and the equality is
achieved when $\bar{\bx}$ is the corresponding eigenvector of $\lambda_{max}(\bar{\bb_t}\bar{\bb_t}^T)$. Since $\bar{\bb_t}\bar{\bb_t}^T$ is rank-one with the eigenvalue $|\bar{\bb_t}|^2$ and the eigenvector $\bar{\bb_t}$, we have $f(\bx) \leq  \bar{\bb_t}^T\bar{\bb_t}= \bb_t^T \bQ^{-1}\bb_t$ and the equality is achieved at $\bar{\bx}^* = \bar{\bb_t}$ or $\bx^* = \bQ^{-1}\bb_t $ and its non-zero scales.
%\vspace{-0.4cm}
%
%=========Appendix B================================================
%
%%\vspace{-0.3cm}
\subsection{Proving that solution of $(\cO_2^l)$ satisfies TPC at the equality}\label{appendix-proof-O3-TPC}
%%\vspace{-0.3cm}
%
Consider $(\cO_2^l)$ and let $\mu$ and $\bpsi$, $\bphi$ be the Lagrange multipliers corresponding to $\ba_t^T \ba_t \! \leq \! {\cal P}_{tot}$, $\ba_t \! \preceq \! \sqrt{\bcP_0}$ and $ \ba_t \! \succeq \! \bzero $, respectively.
%The Lagrangian cost function is ${\cal L}(\ba_t, \mu, \boldsymbol\phi, \boldsymbol\psi)\!= \! \frac{c}{\ba_t^T \bb_t \bb_t^T \ba_t}+\mu ({\cal P}_{tot}-\ba_t^T\ba_t) + \boldsymbol\psi^T (\ba_t - \sqrt{\bcP_0}) -\boldsymbol\phi^T \ba_t$.
The KKT conditions are
\begin{align}
&\frac{-2cb_{t_k} \bb_t^T\ba_t}{|\ba_t^T \bb_t \bb_t^T \ba_t|^2}-2\mu a_{t_k}+\psi_k-\phi_k=0,~~~ k=1,...,M\label{kkt-low-snr-1}\\
&\mu(\ba_t^T\ba_t-{\cal P}_{tot})=0,~~\mu\geq 0,~~\ba_t^T\ba_t\leq {\cal P}_{tot}\label{kkt-low-snr-2}\\
&\psi_k (a_{t_k}-\sqrt{{\cal P}_{0_k}})=0,~~\psi_k \geq 0,~~a_{t_k} \leq \sqrt{{\cal P}_{0_k}},~~\mbox{and}\nonumber\\
&\phi_k a_{t_k}=0,~~\phi_k\geq 0,~~ a_{t_k} \geq 0 \label{kkt-low-snr-4}
\end{align}
We show $\mu \! \neq \!0$.
Substituting $\mu\!=\!0$ in (\ref{kkt-low-snr-1}), we have
$\frac{-2cb_{t_k} \bb_t^T\ba_t}{|\ba_t^T \bb_t \bb_t^T \ba_t|^2}\!=\!\phi_k-\psi_k$.
Since $ \bb_t \! \succ \! \bzero$, $ \ba_t \! \succeq \! \bzero$, $ \ba_t \! \neq \! \bzero$, we find $\frac{-2cb_{t_k} \bb_t^T\ba_t}{|\ba_t^T \bb_t \bb_t^T \ba_t|^2} \!= \! \phi_k-\psi_k \!< \! 0$. Note that $\psi_k$ and $\phi_k$ cannot be both positive,
since from (\ref{kkt-low-snr-4}) it is infeasible to have  $\ba_{t_k}\!=\!\sqrt{\cP_0}_k$ and $\ba_{t_k}\!=\!0$. Therefore $\psi_k$ or $\phi_k$ must be zero.
Since $\phi_k-\psi_k\!<\!0$, we conclude that $\phi_k\!=\!0$ and $\psi_k\!>\!0$. Now, from (\ref{kkt-low-snr-4}) we have $\ba_{t_k}\!=\!\sqrt{{\cal P}_{0_k}}$, leading to  $\ba_t^T\ba_t\!=\!\sum_{k=1}^M \cP_{0_k} \!>\! {\cal P}_{tot}$, which contradicts \eqref{kkt-low-snr-2}.
%to ensure that total and individual power constraints are active.
Therefore, $\mu \! \neq \! 0 $ and we have $\ba_t^T\ba_t\!=\!{\cal P}_{tot}$.
\subsection{Analytical Solution of $(\cO_2'')$}\label{coh-case2}
%\vspace{-1.0cm}
Since $\ba^{*T}_{t1}\ba^{*}_{t1}\!=\!{\cal P}_{tot}$, the objective function in $(\cO_2'')$ reduces to ${\cal P}_{tot} - \ba^{*T}_{t1} \ba_t$. Let $\mu$ and $\bpsi$ and $\bphi$ be the Lagrange multipliers corresponding to  $\ba_t^T\ba_t \!=\! {\cal P}_{tot}$, $\ba_t \! \preceq \!\sqrt{\bcP_0}$ and $ \ba_t \! \succeq \! \bzero$. The KKT conditions are
%
%\vspace{-0.2cm}
\begin{align*}
&2\mu a_{t_k}   - a^*_{t1_k}  + \psi_k - \phi_k  =0,~~~ k=1,...,M  \nonumber\\ %\label{main_eq2}
&\mu(\ba_t^T\ba_t-{\cal P}_{tot})=0,~~\ba_t^T\ba_t= {\cal P}_{tot}\nonumber\\
& \psi_k(a_{t_k} -\sqrt{\cP_{0_k}})=0,~~ \psi_k\geq 0,~~a_{t_k} \leq \sqrt{\cP_{0_k}}, ~~\mbox{and}\nonumber\\
&\phi_k a_{t_k} = 0,~~\phi_k\geq 0,~~a_{t_k} \geq 0 %\label{zero_ineq2}
\end{align*}
%%\vspace{-0.2cm}
%
%===============================================================
%AZADEH: I understood the argument below, however, due to lack of space, I suggest to remove it.
%
%We note that $a^*_{t1_k} \geq 0$. For $a^*_{t1_k} > 0$ we have these options: (a) $a_{t_k} =0$: Following \eqref{pi_ineq2}, $\psi_k = 0$. If we substitute these results in \eqref{main_eq2}, since $\phi_k\geq 0$ from \eqref{zero_ineq2}, we see that this case never happens. (b) $a_{t_k} = \sqrt{\cP_{0_k}}$: From \eqref{zero_ineq2}, we have $\phi_k = 0$. If we substitute them in \eqref{main_eq2}, to keep $\psi_k\geq 0$ we should have $\mu \leq \frac{a^*_{t1_k}}{2 \sqrt{\cP_{0_k}} }$. (c) $0 < a_{t_k} < \sqrt{\cP_{0_k}}$: In this case $\phi_k = \psi_k = 0$. If we apply them to \eqref{main_eq2}, we reach to $a_{t_k} = \frac{a^*_{t1_k}}{2\mu}$. To satisfy $0 < a_{t_k} < \sqrt{\cP_{0_k}}$ we should have $\mu > \frac{a^*_{t1_k}}{2 \sqrt{\cP_{0_k}}}$. We summarize this discussion as
%
%=========================================================================
%
Solving the KKT conditions for $\ba^*_{t1} \! \succ \! \bzero$ yields
\begin{equation}\label{case2-a}
a_{t_k} = \left\{
\begin{array}{ll}
                   \frac{a^*_{t1_k}}{2\mu}, & ~ \text{for}~\mu \geq \frac{a^*_{t1_k}}{2 \sqrt{\cP_{0_k}}}\\
                   \sqrt{\cP_{0_k}}, &~ \text{for}~\mu < \frac{a^*_{t1_k}}{2 \sqrt{\cP_{0_k}}}. \\
                                    \end{array}\right.
\end{equation}
To find positive $\mu$ and consequently $a_{t_k}$, suppose sensors are sorted such that $\frac{a_{t1_{i_1}}^*}{\sqrt{\mathcal{P}_{0_{i_1}}}} \! \geq \! ...\! \geq \!\frac{a_{t1_{i_M}}^*}{\sqrt{\mathcal{P}_{0_{i_M}}}}$. Assume for $1 \leq m \leq M $ we have $a_{t_{i_1}}\!=\!\sqrt{{\cP}_{0_{i_1}}},...,a_{t_{i_m}}\!=\!\sqrt{{\cP}_{0_{i_m}}}$. Substituting (\ref{case2-a}) into $\sum_{j=1}^M a_{t_{i_j}}^2\!=\! {\cal P}_{tot}$ and solving for $\mu$, we find
%
%we find $\sum_{j=1}^m {\cP}_{0_{i_j}} \! + \! \frac{1}{4\mu^2}\sum_{j=m+1}^M a^{*2}_{t1_{i_j}} \! =\! {\cal P}_{tot}$ which results in
%
$\mu^2=\frac{\sum_{j=m+1}^M a^{*2}_{t1_{i_j}}}{4({\cP}_{tot}-\sum_{j=1}^{m} {\cP}_{0_{i_j}})}$. Note that $\mu$ depends on $m$. If $\frac{a_{t1_{i_{m+1}}}^*}{2\sqrt{{\cP}_{0_{i_{m+1}}}}} \leq \mu \leq \frac{a_{t1_{i_m}}^*}{2\sqrt{{\cP}_{0_{i_m}}}}$, the above assumption is valid, and we substitute $\mu$ in (\ref{case2-a}) to calculate $a_{t_{i_{m+1}}},...,a_{t_{i_M}}$. Otherwise, we increase $m$ by one and repeat the procedure, until we reach $\mu$ that lies within the proper interval.
Although $({\cO}_3'')$ is not convex, we show below that the KKT solution in (\ref{case2-a}) is unique.
Suppose the solution in (\ref{case2-a}) is not unique, i.e.,  there exist $1 \leq m, m' \leq M$, $m'  \! \geq \! m+1$ such that
\begin{eqnarray}\label{case2-c}
a_{t_{i_j}} =  \left\{
\begin{array}{ll}
\frac{a_{t1_{i_j}}^*}{2\mu}, & \quad\text{for } m+1 \leq j \leq M, \mbox{and}~ \mu\geq \frac{a^{*}_{t1_{i_j}}}{2\sqrt{{\cP}_{0_{i_j}}}}, \\
\sqrt{\mathcal{P}_{0_{i_j}}}, & \quad\text{for } 1 \leq j \leq m, \mbox{and}~ \mu< \frac{a^{*}_{t1_{i_j}}}{2\sqrt{{\cP}_{0_{i_j}}}}
\end{array}\right.
\end{eqnarray}
Also, $a_{t_{i_{j'}}}$ can be obtained  from (\ref{case2-c}) by substituting $j,m,\mu$ with $j',m',\mu'$, respectively.
%======Azadeh: this is repeatition, I replaced it with text=========================
%\begin{eqnarray}\label{case2-d}
%\mu'^2=\frac{\sum_{j=m'+1}^M a^{*2}_{t1_{i_j}}}{4({\cP}_{tot}-\sum_{j=1}^{m'} {\cP}_{0_{i_j}})}~\text{and}~a_{t_{i_j}} = \left\{
%\begin{array}{ll}
%\frac{a_{t1_{i_j}}^*}{2\mu'}, & \quad\text{for } j \! \in \!  \{m'+1,...M\}, \mbox{where}~ \mu'\geq \frac{a^{*}_{t1_{i_j}}}{2\sqrt{{\cP}_{0_{i_j}}}},\\
%\sqrt{\mathcal{P}_{0_{i_j}}}, & \quad\text{for }  j \! \in \! \{1,...m'\}, \mbox{where}~ \mu'< \frac{a^{*}_{t1_{i_j}}}{2\sqrt{{\cP}_{0_{i_j}}}}
%\end{array}\right.
%\end{eqnarray}
%====================================================================
%
Since $m' \! \geq \! m+1$, from (\ref{case2-c}), we have $\mu^2\geq\frac{a^{*2}_{t1_{i_{m+1}}}}{4{\cP}_{0_{i_{m+1}}}}$.
On the other hand, due to sensor ordering, we have $\frac{a^{*2}_{t1_{i_{m+1}}}}{4{\cP}_{0_{i_{m+1}}}}\geq...\geq\frac{a^{*2}_{t1_{i_{m'}}}}{4{\cP}_{0_{i_{m'}}}}$. Applying the mediant inequality, we obtain
$\frac{\sum_{j=m+1}^{m'}a^{*2}_{t1_{i_j}}}{4\sum_{j=m+1}^{m'}{\cP}_{0_{i_j}}}\geq\frac{a^{*2}_{t1_{i_{m+1}}}}{4{\cP}_{0_{i_{m+1}}}}$.
We observe that two different fractions are greater than or equal to the $\frac{a^{*2}_{t1_{i_{m+1}}}}{4{\cP}_{0_{i_{m+1}}}}$. Hence, using the mediant inequality and definition of $\mu^2$, we find
%
%%\vspace{-0.2cm}
\begin{align*}
\frac{\sum_{j=m+1}^M a^{*2}_{t1_{i_j}}-\sum_{j=m+1}^{m'} a^{*2}_{t1_{i_j}}}{4({\cP}_{tot}-\sum_{j=1}^{m} {\cP}_{0_{i_j}}-\sum_{j=m+1}^{m'} {\cP}_{0_{i_j}})}=\mu'^2 \geq \frac{a^{*2}_{t1_{i_{m+1}}}}{4{\cP}_{0_{i_{m+1}}}}.
\end{align*}
%%\vspace{-0.1cm}
%
However, this inequality contradicts the one in (\ref{case2-c}) when $j,\mu$ are replaced with $j',\mu'$. Hence, our assumption regarding the existence of $m,m'$ is incorrect and the solution in (\ref{case2-a}) is unique.
%================Appendix F==========================================================
%
%%\vspace{-0.2cm}
\subsection{Proving that solution of $(\cO_3^l)$ is equal to IPC upper limit}\label{appendix-proof-O5-IPC}
%%\vspace{-0.2cm}
%
Consider $(\cO_3^l)$ and let  $\bpsi$, $\bphi$ be the Lagrange multipliers corresponding to  $\ba_t \! \preceq \! \sqrt{\bcP_0}$ and $ \ba_t \! \succeq \! \bzero$, respectively.
%The Lagrangian cost function is ${\cal L}(\ba_t,\boldsymbol\phi, \boldsymbol\psi)\! = \! \frac{c}{\ba_t^T \bb_t \bb_t^T \ba_t}+ \boldsymbol\psi^T (\ba_t - \sqrt{\bcP_0}) -\boldsymbol\phi^T \ba_t$.
The KKT conditions are
\begin{align}
&\frac{-2cb_{t_k} \bb_t^T\ba_t}{|\ba_t^T \bb_t \bb_t^T \ba_t|^2}+\psi_k-\phi_k=0,~~~ k=1,...,M\nonumber\\ %\label{kkt-low-snr-5}
&\psi_k (a_{t_k}-\sqrt{{\cal P}_{0_k}})=0,~~\psi_k \geq 0,~~a_{t_k} \leq \sqrt{{\cal P}_{0_k}},\nonumber\\
&\mbox{and}~~~
\phi_k a_{t_k}=0,~~\phi_k \geq 0,~~ a_{t_k} \geq 0 \label{kkt-low-snr-7}
\end{align}
Since $ \bb_t \! \succ \! \bzero$, $ \ba_t \! \succeq \! \bzero$, $ \ba_t \! \neq  \! \bzero$, we find $\frac{-2cb_{t_k} \bb_t^T\ba_t}{|\ba_t^T \bb_t \bb_t^T \ba_t|^2} \!= \!\phi_k-\psi_k \!< \!0$. Note that $\psi_k$ and $\phi_k$ cannot be both positive, since  from (\ref{kkt-low-snr-7}) it is infeasible to have $\ba_{t_k}\!=\!\sqrt{\cP_0}_k$  and $\ba_{t_k}\!=\!0$. Therefore either $\psi_k$ or $\phi_k$ must be zero. Since $\phi_k-\psi_k\!<\!0$, we conclude that $\phi_k\!=\!0$ and $\psi_k\!>\!0$. Now, from (\ref{kkt-low-snr-7}) we have  $\ba_{t_k}\!=\!\sqrt{{\cal P}_{0_k}}$, or equivalently, $\ba_t \! = \! \sqrt{\bcP_0}$.
%
%=========================Appendix G==================================================
%
%\vspace{-0.4cm}
\subsection{Regarding Analytical Solution to $(\cO_3)$ with Independent Observations}\label{uniqueness}
%\vspace{-0.2cm}
%
First, we show that at least one of ${a}_{t_k}$'s in (\ref{a-tk}) is equal to $\sqrt{\cP_{0}}$.
Suppose ${a}_{t_k}\!<\!\sqrt{\cP_{0}},\forall k$. From \eqref{a-tk}, we have ${a}_{t_k} \!= \! \frac{b_{t_k}}{[\bK_t]_{kk}} \eta $ or equivalently  $\eta  \!= \!\frac{a_{t_k}[\bK_t]_{kk}{a}_{t_k}}{b_{t_k}{a}_{t_k}}, \forall k$.
%==========due to lack of space, I omitted below===================================================
%Recall mediant inequality which states for $M$ ordered fractions $\frac{\alpha_1}{\beta_1} \! \geq \! \frac{\alpha_2}{\beta_2}\! \geq \!... \! \geq \!\frac{\alpha_M}{\beta_M}$ with non negative $\alpha_k$ and positive $\beta_k$ we have $\frac{\alpha_M}{\beta_M}  \! \leq \!   \frac{\sum_{k=1}^M \alpha_k}{\sum_{k=1}^M \beta_k}   \leq \!\frac{\alpha_1}{\beta_1}$. For the special case where the fractions are equal, they are also identical to their mediant.
%========================================================================================
Using the mediant inequality, we rewrite $\eta$ as
%
%\begin{equation*}
$\eta \!= \! \frac{\sum_{k=1}^M {a}_{t_k}[\bK_t]_{kk} {a}_{t_k}}{\sum_{k=1}^M b_{t_k}{a}_{t_k}} \! = \! \frac{{\ba}_t^{T} \bK_t {\ba}_t}{\bb_t^T{\ba}_t} \! < \! \frac{{\ba}_t^T \bK_t {\ba}_t + c}{\bb_t^T {\ba}_t}$,
%\end{equation*}
%
since $c\!>\!0$.
However, this violates the definition of $\eta$ in (\ref{eta-definition}) and contradicts our initial assumption. Hence, there should be at least one ${a}_{t_k} \!= \! \sqrt{\cP_{0}}$.
Next, we show that, although $(\cO_3)$ is not convex, the KKT solution in (\ref{a-tk}) is unique.
Suppose sensors are sorted such that
$\frac{b_{t_{i_1}}}{[\bK_t]_{i_1 i_1}} \! \geq \! ... \! \geq \!\frac{b_{t_{i_M}}}{[\bK_t]_{i_M i_M}}$, i.e., $a_{t_{i_1}} \! \geq \! ... \! \geq \! a_{t_{i_M}}$, where at least  $a_{t_{i_1}} \!=\! \sqrt{\cP_{0}}$.
%Let ${\ba}_t\!=\![a_{t_{i_1}},..., a_{t_{i_m}}, a_{t_{i_{m+1}}},..., a_{t_{i_M}}]^T$. According to Appendix  \ref{uniqueness} at least $a_{t_{i_1}} \!=\! \sqrt{\cP_{0}}$. Assume for $m \! \in \! \{1, 2, ... , M\}$ we have $a_{t_{i_1}} \!= \! ... \!= \! a_{t_{i_m}} \!= \! \sqrt{\cP_{0}}$,
Assume that the solution in (\ref{a-tk}) is not unique, i.e., there exist two indices $m, m' \! \in \!\{1,...,M\}$, $m' \! \geq \! m+1$ such that
%
%\vspace{-0.2cm}
\begin{eqnarray*}\label{a-tk-a}
a_{t_{i_j}} = \left\{
\begin{array}{ll}
\frac{b_{t_{i_j}}}{  [\bK_t]_{i_j i_j}    } \eta, & \quad\text{for } m+1 \leq j \leq M,~ \eta \! \leq \!\frac{   [\bK_t]_{i_j i_j}   }{ b_{t_{i_j}}   } \sqrt{\cP_{0}}\\
\sqrt{\cP_{0}}, & \quad\text{for }  1 \leq j \leq m, \mbox{and}~ \eta \! > \!\frac{   [\bK_t]_{i_j i_j}   }{ b_{t_{i_j}}   } \sqrt{\cP_{0}}
\end{array}\right.
\end{eqnarray*}
%%\vspace{-0.1cm}
%
where $\eta \!= \! \frac{\cP_{0} \sum_{\ell=1}^{m}  [\bK_t]_{i_\ell i_\ell} + c}{\sqrt{\cP_{0}} \sum_{\ell=1}^{m}b_{t_{i_\ell}}}$. Also, $a_{t_{i_{j'}}}$ can be obtained  from above by substituting $j,m,\eta$ with $j',m',\eta'$, respectively.
Since $m' \! \geq \! m+1$, from  (\ref{a-tk-a}), we find
%\begin{align}\label{inequality-rho}
$\eta \! \leq \! \frac{[\bK_t]_{i_{m'}i_{m'}}}{b_{t_{i_{m'}}}}\sqrt{\cP_{0}}$.
%\end{align}
On the other hand, due to sensor ordering, we have $\frac{b_{t_{i_{m+1}}}}{[\bK_t]_{i_{m+1} i_{m+1}}} \! \geq \! ... \! \geq \!\frac{b_{t_{i_{m'}}}}{[\bK_t]_{i_{m'} i_{m'}}}$. Applying the mediant inequality, we obtain
%
%$\frac{b_{t_{i_{m'}}}}{[\bK_t]_{{i_{m'}i_{m'}}}} \! \leq \! \frac{\sum_{\ell=m+1}^{m'}b_{t_{i_\ell}}}{\sum_{\ell=m+1}^{m'}[\bK_t]_{{i_\ell  i_\ell}}} ~\Rightarrow \frac{ {\cP_{0}} \sum_{\ell=m+1}^{m'}[\bK_t]_{{i_\ell  i_\ell}}}{ \sqrt{\cP_{0}}\sum_{\ell=m+1}^{m'}b_{t_{i_\ell}}}  \! \leq \! \frac{[\bK_t]_{{i_{m'}i_{m'}}}}{b_{t_{i_{m'}}}}\sqrt{\cP_{0}}$.
%
$ \frac{ {\cP_{0}} \sum_{\ell=m+1}^{m'}[\bK_t]_{{i_\ell  i_\ell}}}{ \sqrt{\cP_{0}}\sum_{\ell=m+1}^{m'}b_{t_{i_\ell}}}  \! \leq \! \frac{[\bK_t]_{{i_{m'}i_{m'}}}}{b_{t_{i_{m'}}}}\sqrt{\cP_{0}}$.
We observe that two different fractions are less than or equal to $\frac{[\bK_t]_{i_{m'}i_{m'}}}{b_{t_{i_{m'}}}}\sqrt{\cP_{0}}$. Thus %using the mediant inequality we find
\begin{align*}
&\frac{\cP_{0} \sum_{\ell=1}^{m} [\bK_t]_{{i_\ell i_\ell}}+ c + \cP_{0}\sum_{\ell=m+1}^{m'} [\bK_t]_{{i_\ell i_\ell}}}{\sqrt{\cP_{0}} \sum_{\ell=1}^{m}b_{t_{i_\ell}}+\sqrt{\cP_{0}}\sum_{\ell=m+1}^{m'}b_{t_{i_\ell}}} \!= \! \eta' \leq\\ &\frac{[\bK_t]_{i_{m'}i_{m'}}}{b_{t_{i_{m'}}}}\sqrt{\cP_{0}}.
\end{align*}
However, this inequality contradicts the one in (\ref{a-tk-a}) when $j,\eta$ are replaced with $j',\eta'$. Hence, our assumption regarding the existence of two indices $m,m'$ is incorrect and the solution in (\ref{a-tk}) is unique.
\vspace{-0.5cm}
\bibliographystyle{IEEETran}
\normalem
\bibliography{Refs}
%%%%%%%%%%%%%%%%%%%%%%%%%%%%%%%
\iftrue
\newpage
\onecolumn
%
%===============LRT vs Linear, Optimal rule vs. MDC===================
\begin{figure}[t]
\centering
\subfigure[]{
\includegraphics[width=5.0in,height=3.0in]{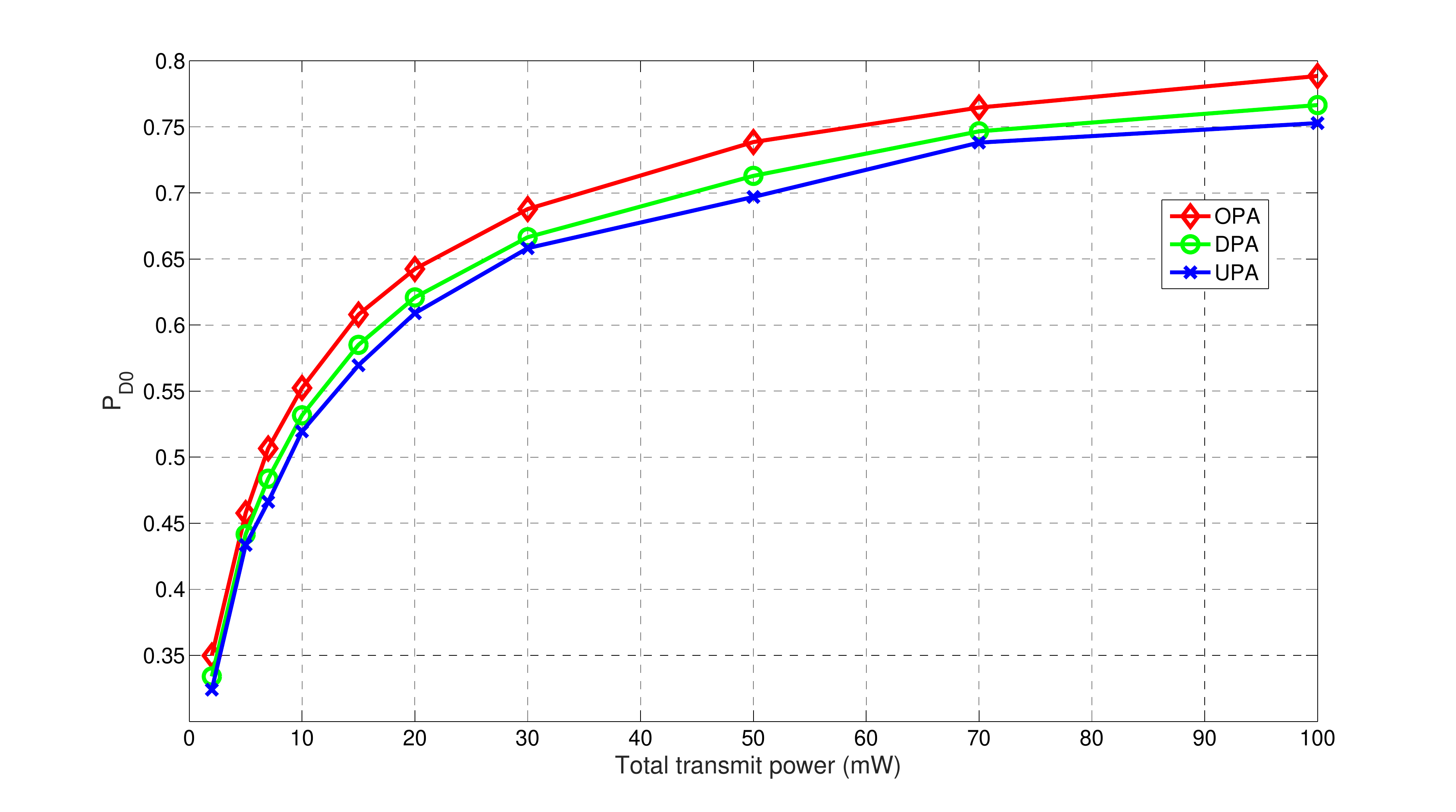}
\label{fig:Linear_2sensor}
}
\subfigure[]{
\includegraphics[width=5.0in,height=3.0in]{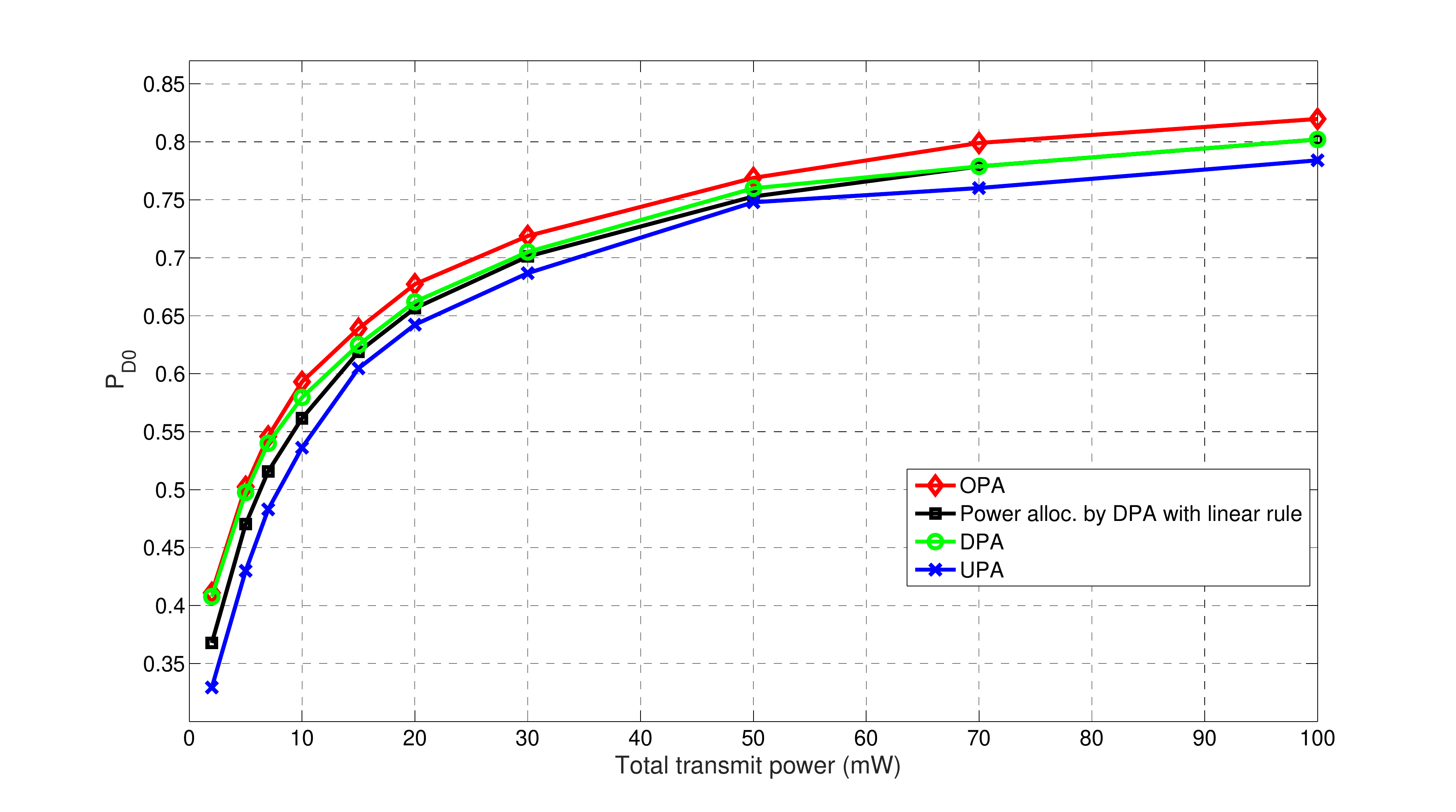}
\label{fig:LRT_2sensor}
}
\caption{$P_{D_0}$ under TPC versus ${\cal P}_{tot}$ for a 2-sensor PAC with identical $p_{d_k}$'s and pathloss and $\rho \!=\!0.1$: (a) Linear fusion rule, (b) LRT fusion rule.}
\label{fig:LRT_Linear}
\end{figure}
%======================================================================
%
%========performance of DPA for coherent (PAC and MAC)============
\begin{figure}[b]
\centering
\includegraphics[width=120mm]{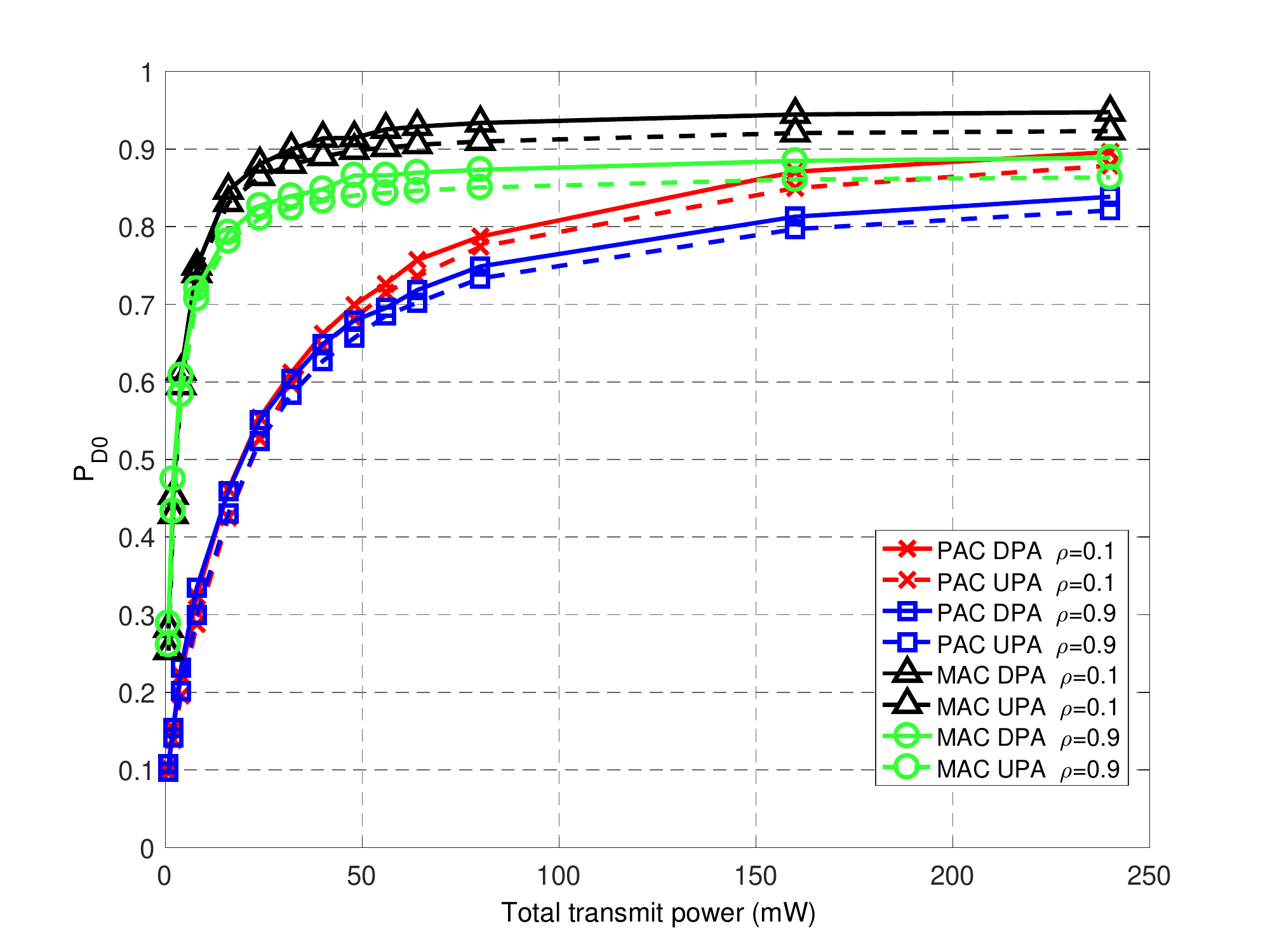}
\caption{$P_{D_0}$ under TPC versus ${\cal P}_{tot}$.}
\label{fig:PD0-TPC-coh}
\end{figure}
\begin{figure}
\centering
\includegraphics[width=120mm]{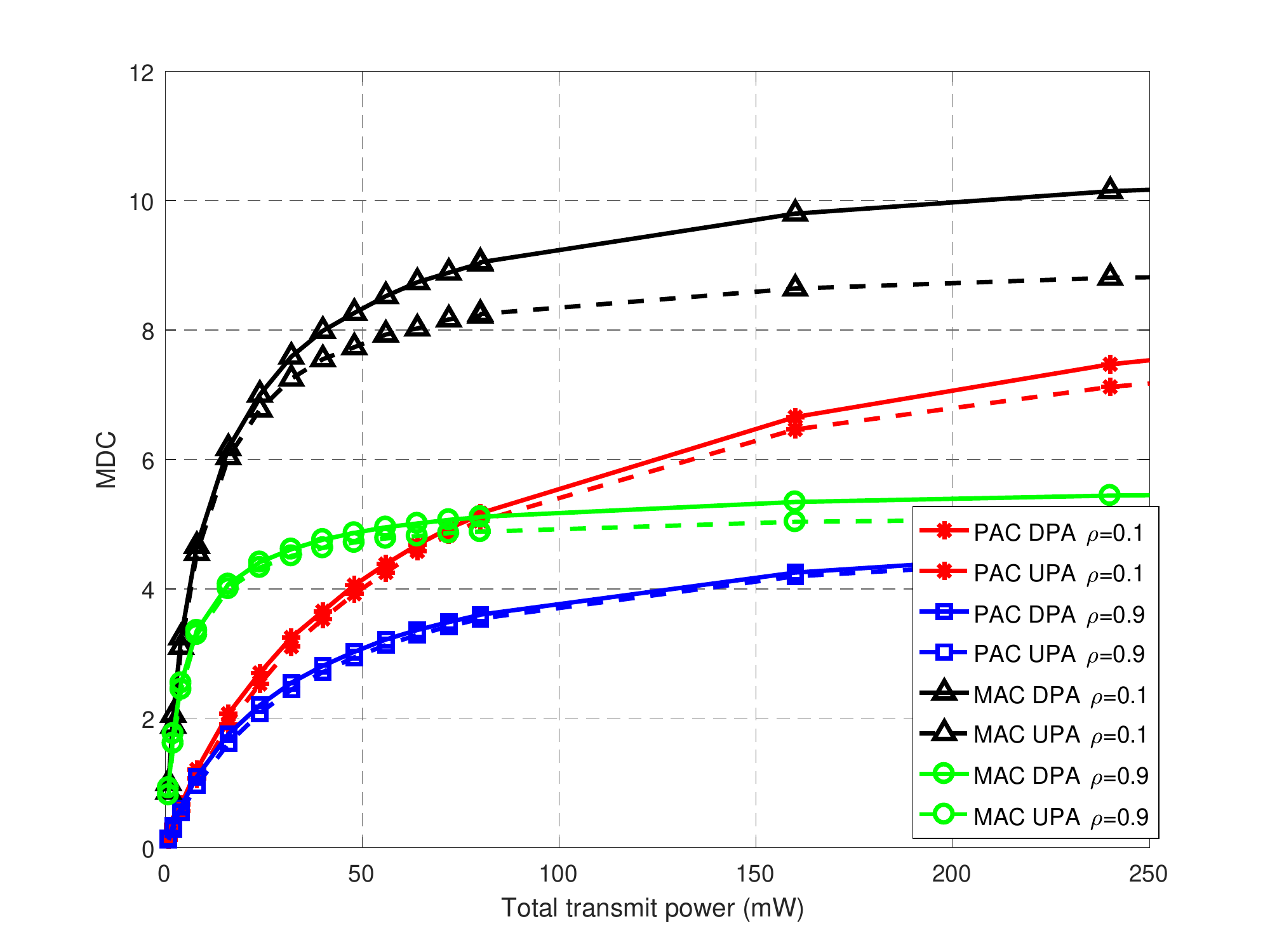}
\caption{Maximized MDC under TPC versus ${\cal P}_{tot}$.}
\label{fig:MDC-TPC-coh}
\end{figure}
\begin{figure}
\centering
\includegraphics[width=120mm]{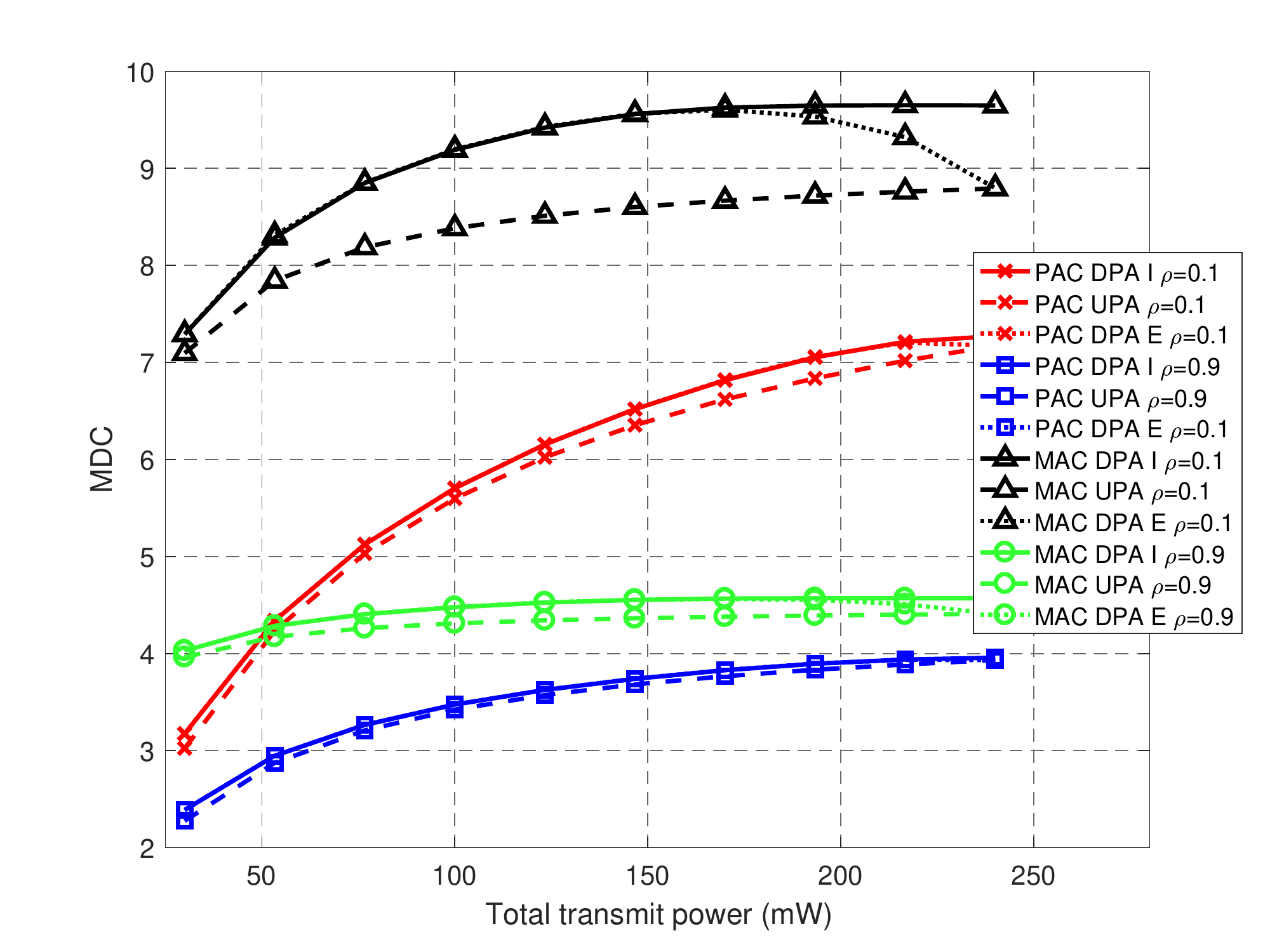}
\caption{Maximized MDC under TIPC versus ${\cal P}_{tot}$ and $\bar{\cP} = 30\, \mW$.}
\label{fig:MDC-TIPC-coh}
\end{figure}
\begin{figure}
\centering
\includegraphics[width=120mm]{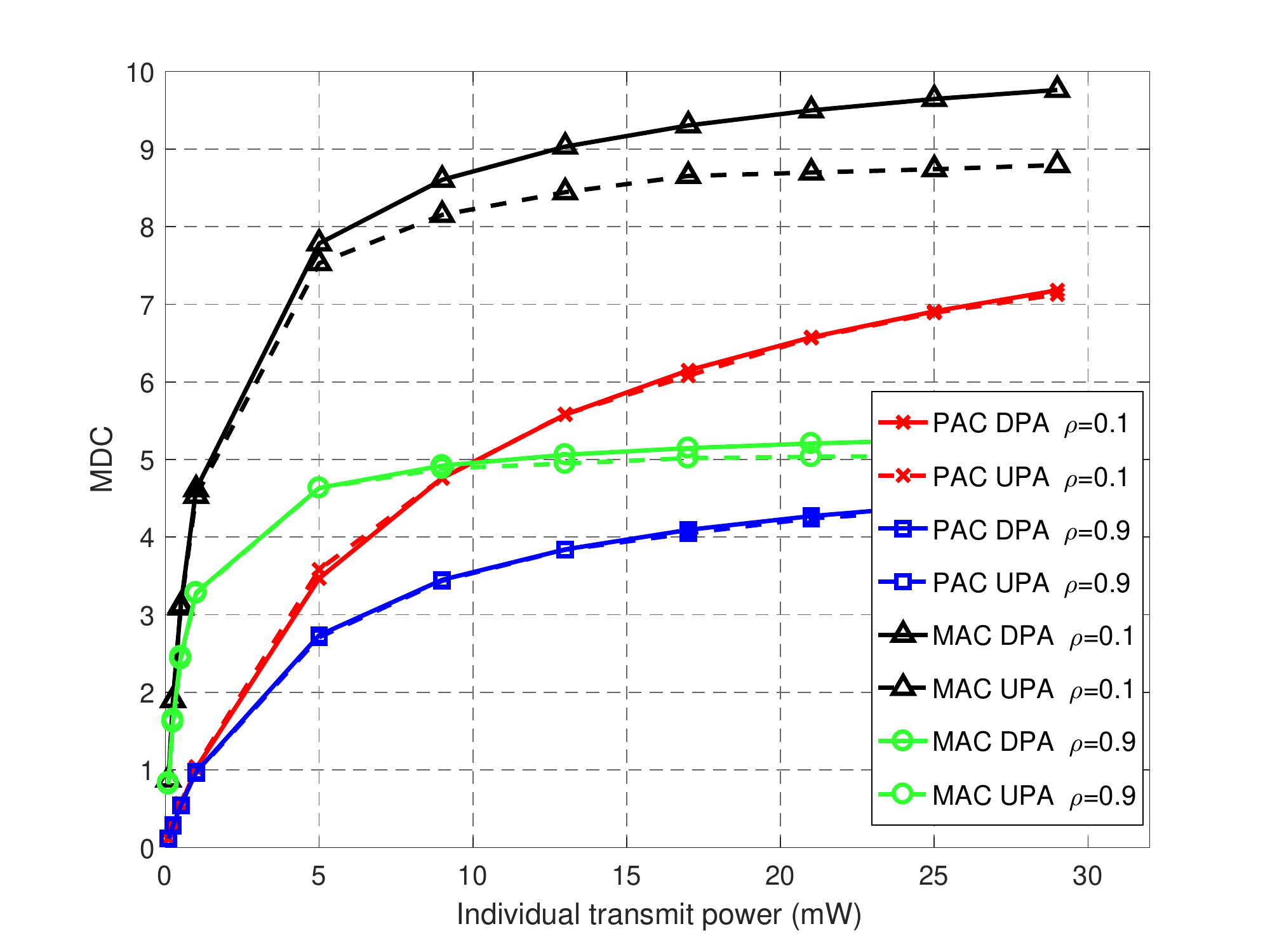}
\caption{Maximized MDC under IPC versus $\bar{\cP}$.}
\label{fig:MDC-IPC-coh}
\end{figure}
%===========Trends of DPA when pdk across sensors are different===========================
\begin{figure}[b]
\centering
\subfigure[]{
\includegraphics[width=2.8in]{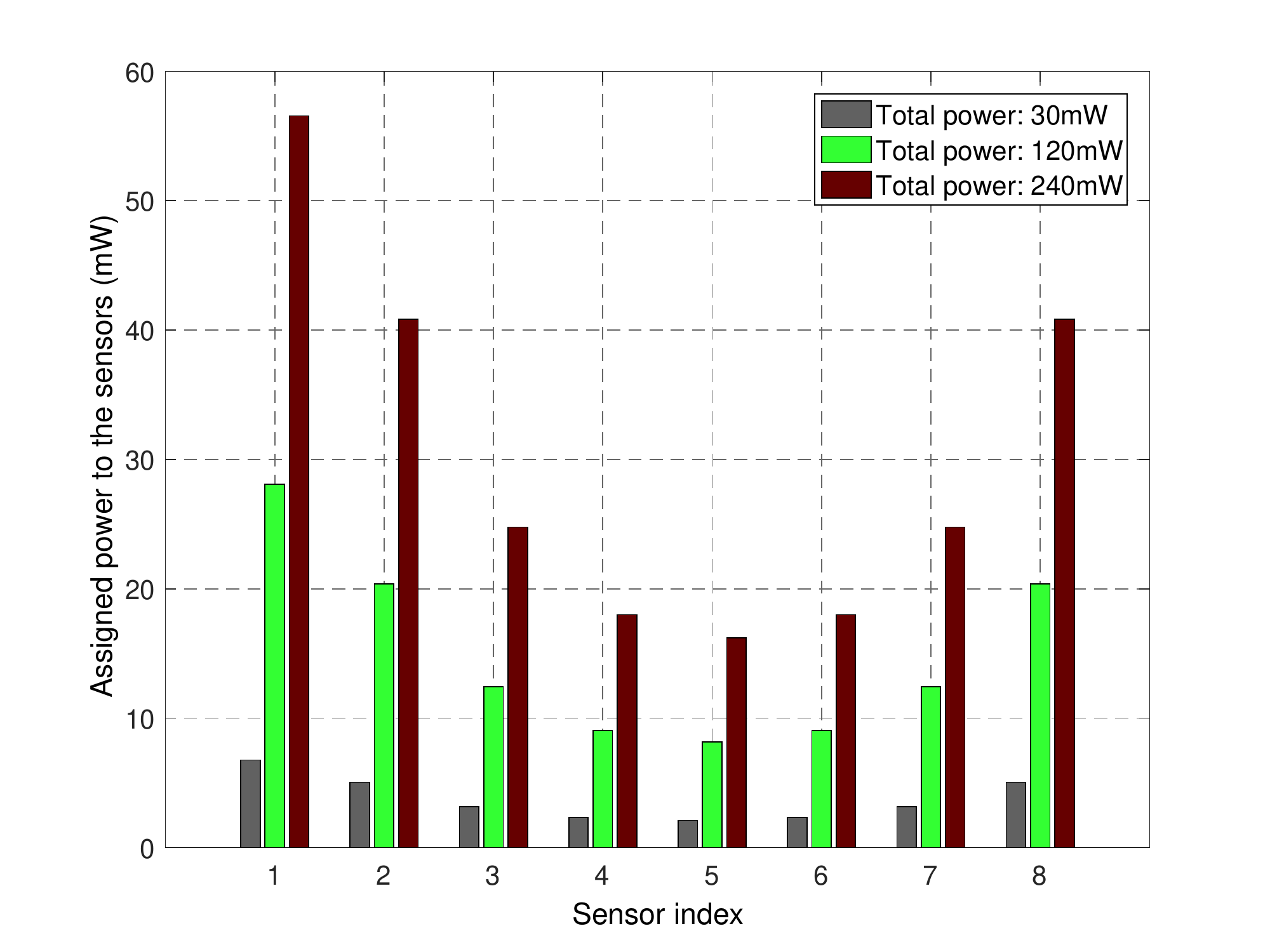}
\label{fig:coh_MAC_case1_01}
}
\subfigure[]{
\includegraphics[width=2.8in]{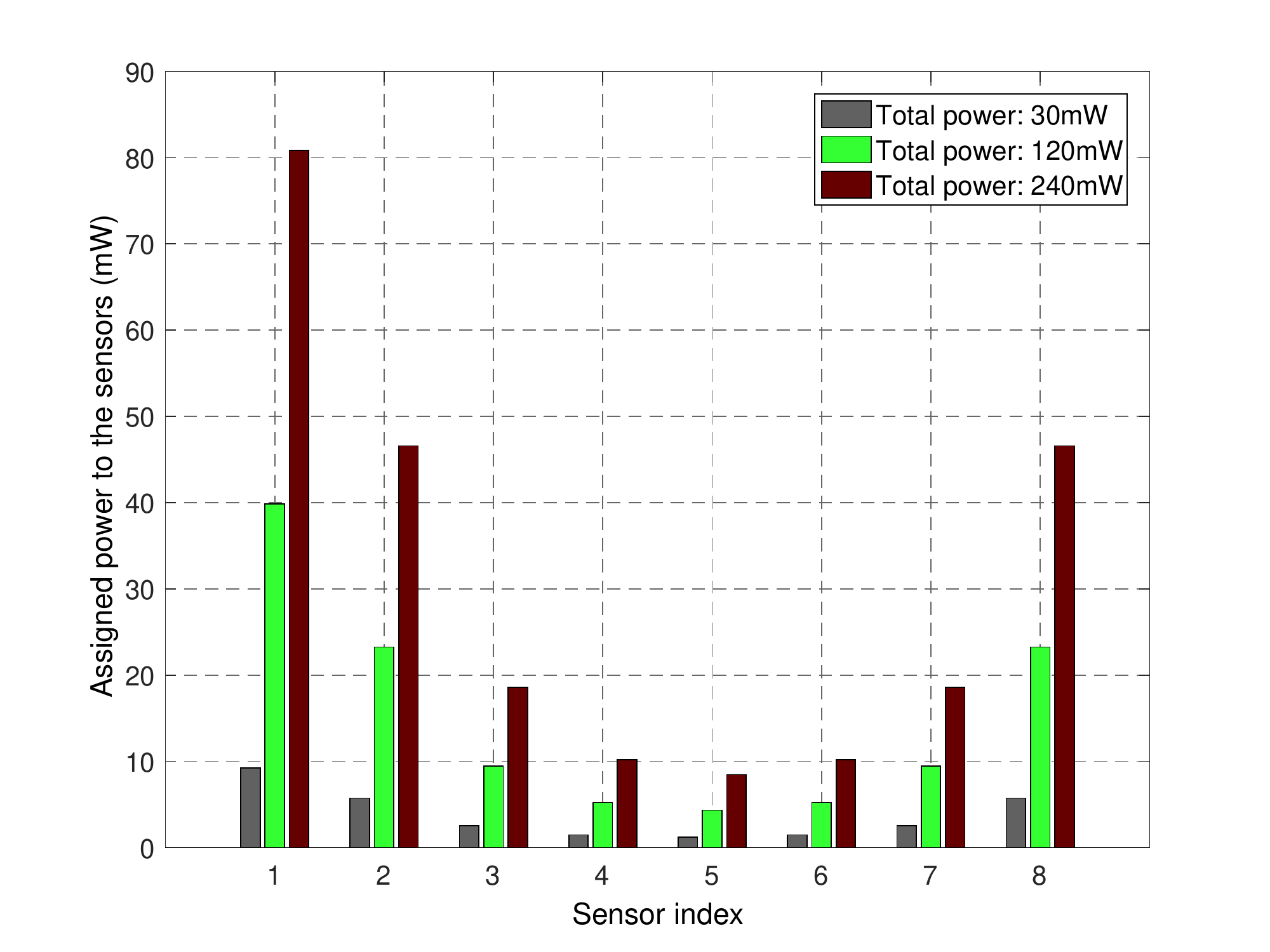}
\label{fig:coh_MAC_case1_09}
}
\subfigure[]{
\includegraphics[width=2.8in]{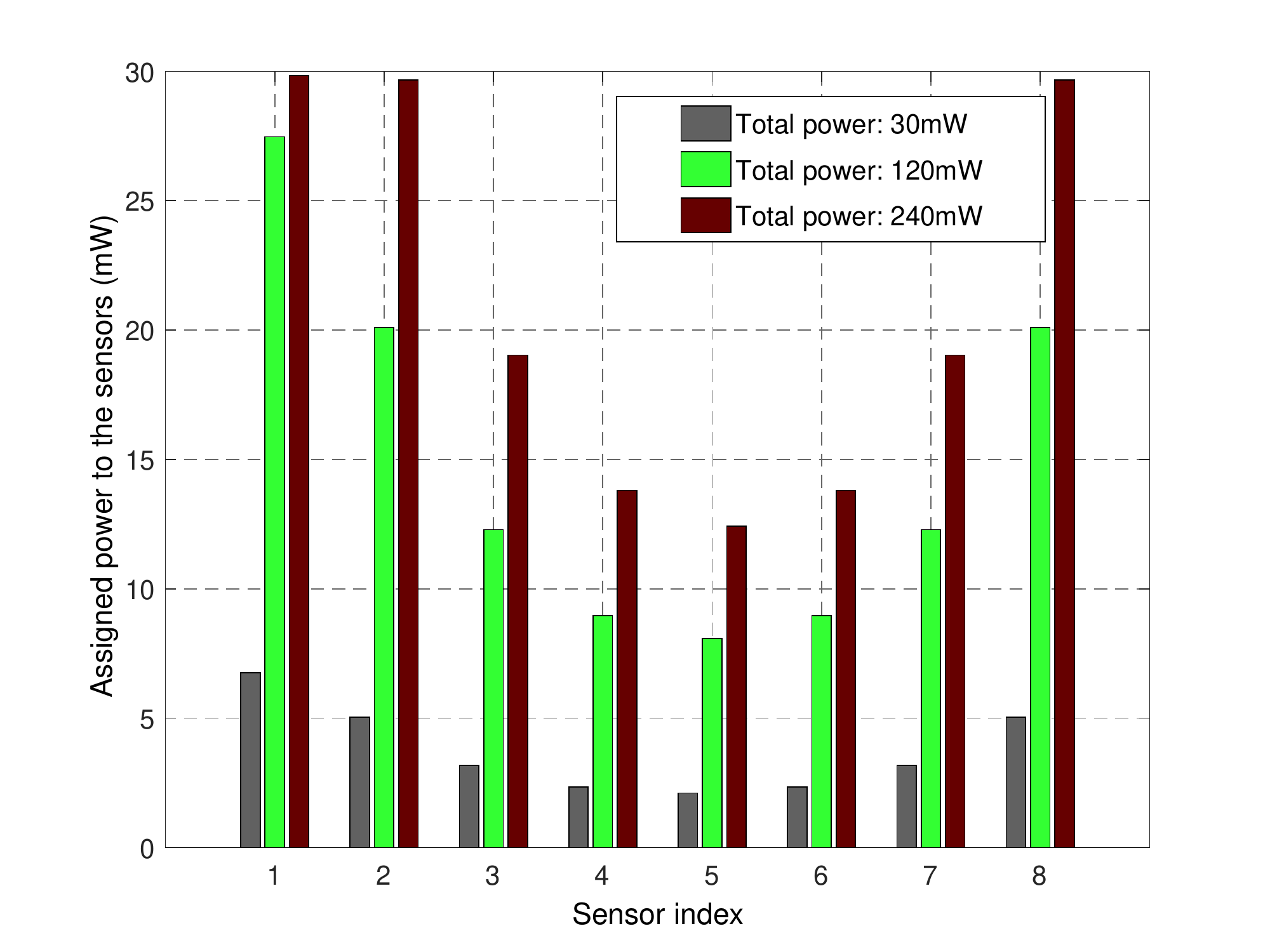}
\label{fig:coh_MAC_case2_01}
}
\subfigure[]{
\includegraphics[width=2.8in]{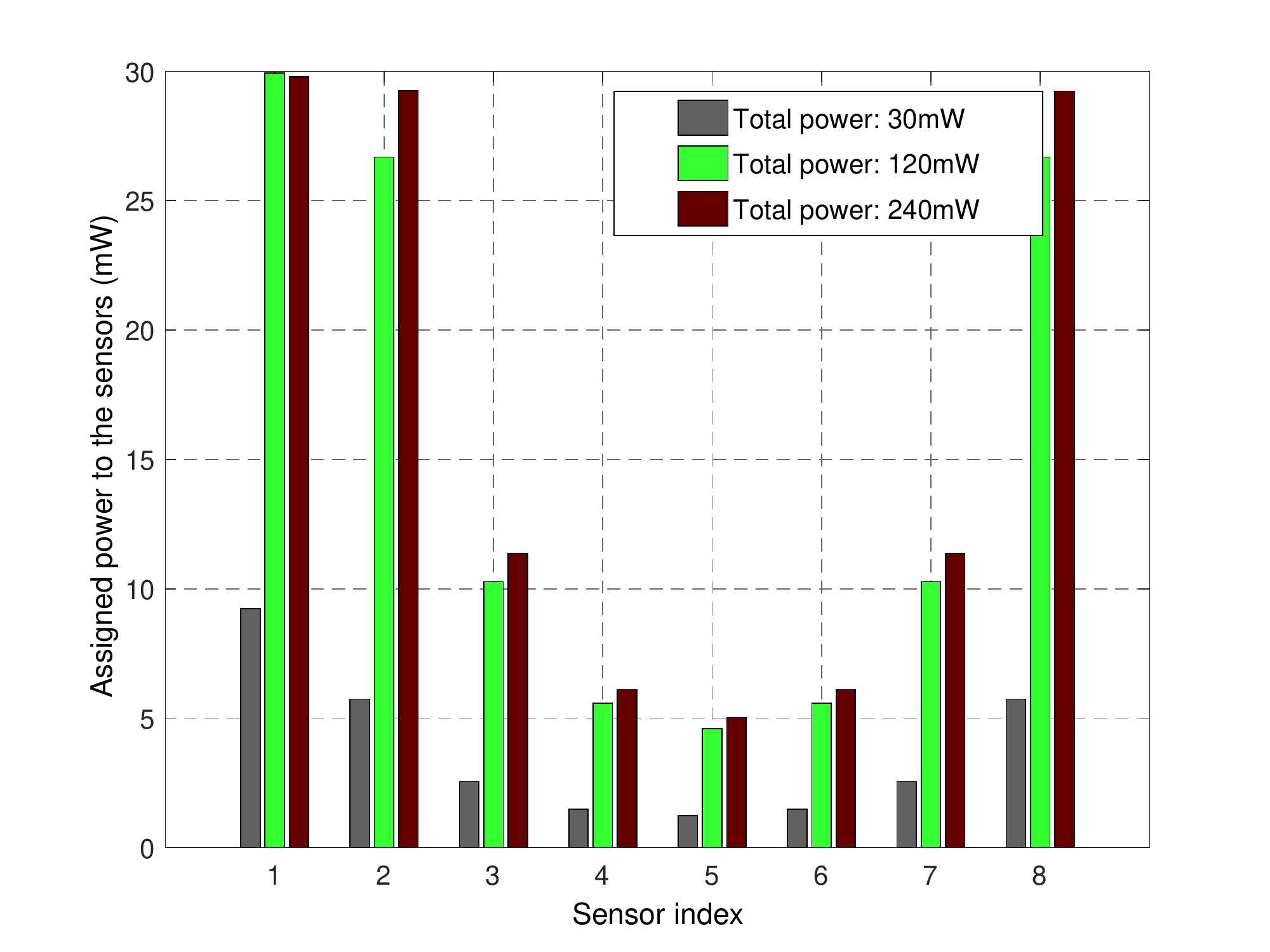}
\label{fig:coh_MAC_case2_09}
}
\subfigure[]{
\includegraphics[width=2.8in]{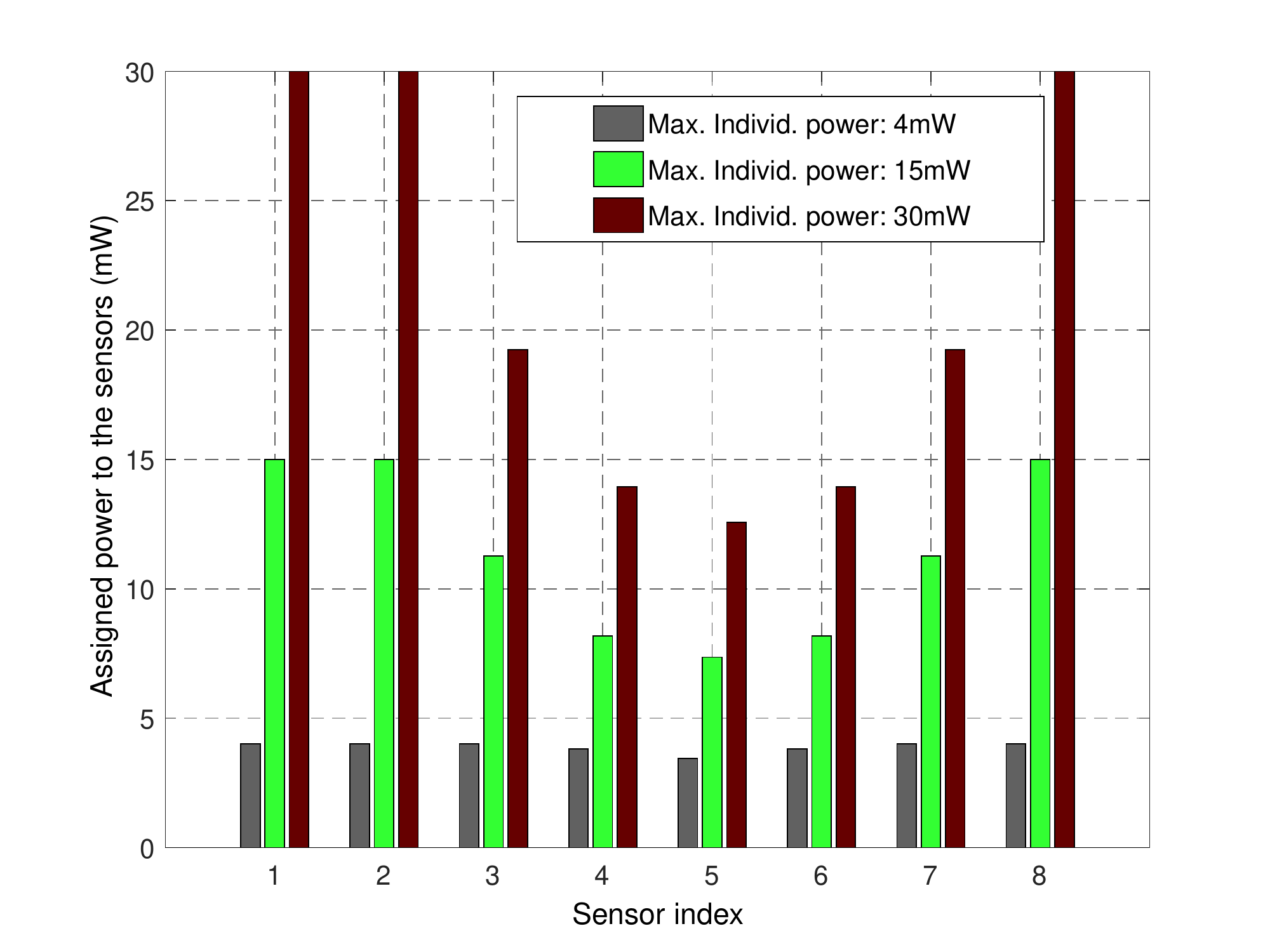}
\label{fig:coh_MAC_case3_01}
}
\subfigure[]{
\includegraphics[width=2.8in]{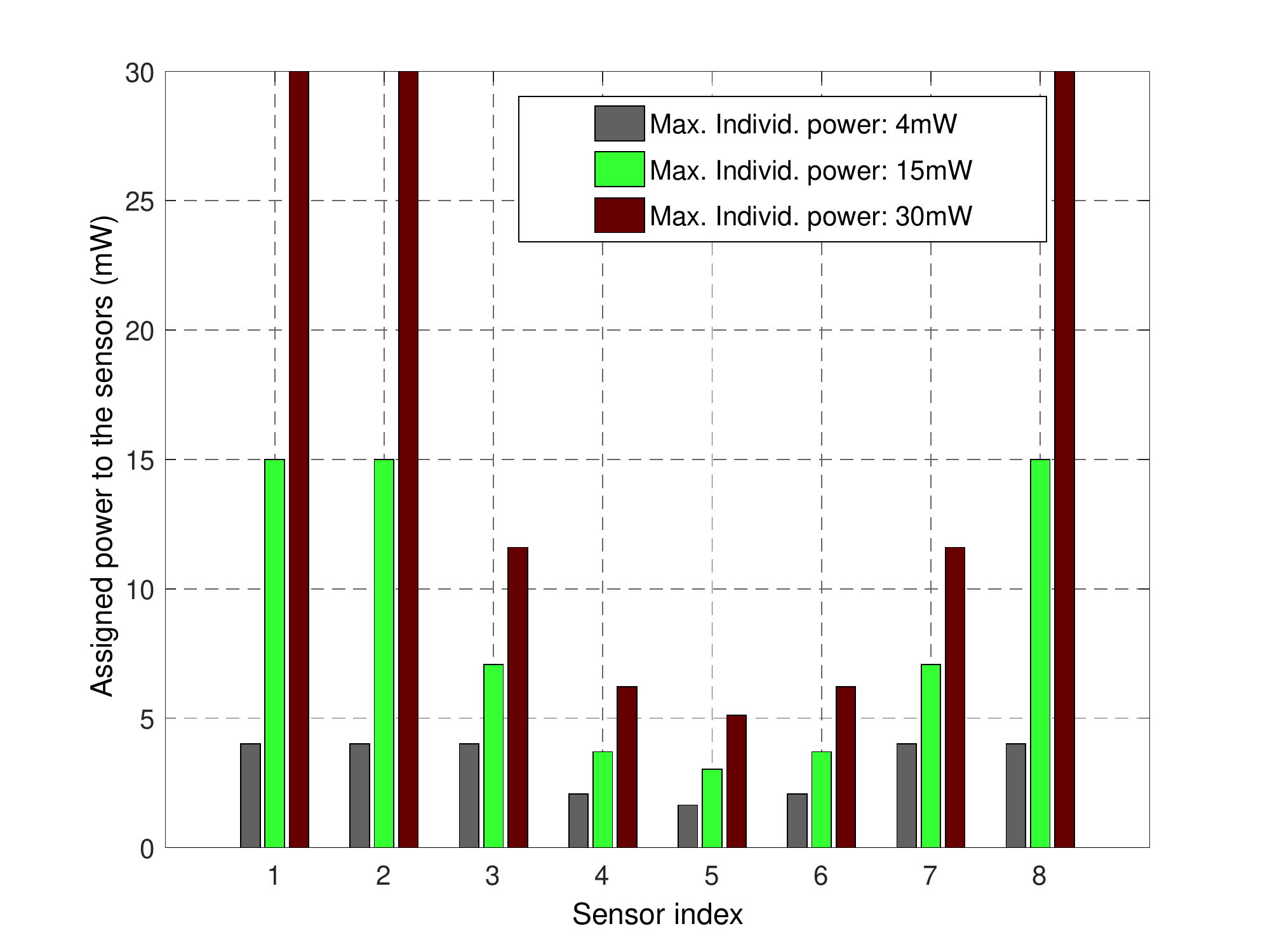}
\label{fig:coh_MAC_case3_09}
}
\caption{DPA in MAC with different $p_{d_k}$'s and identical pathloss: (a) Maximized MDC under TPC,  $\rho =0.1$; (b) Maximized MDC under TPC, $\rho =0.9$; (c) Maximized MDC under TIPC, $\rho =0.1$, $\bar{\cP} = 30\, \mW$, (d)  Maximized MDC under TIPC, $\rho =0.9$, $\bar{\cP} = 30\, \mW$; (e) Maximized MDC under IPC, $\rho =0.1$, (f) Maximized MDC under IPC, $\rho =0.9$.}
\label{fig:coh_MAC}
\end{figure}
\begin{figure}[b]
\centering
\subfigure[]{
\includegraphics[width=2.8in]{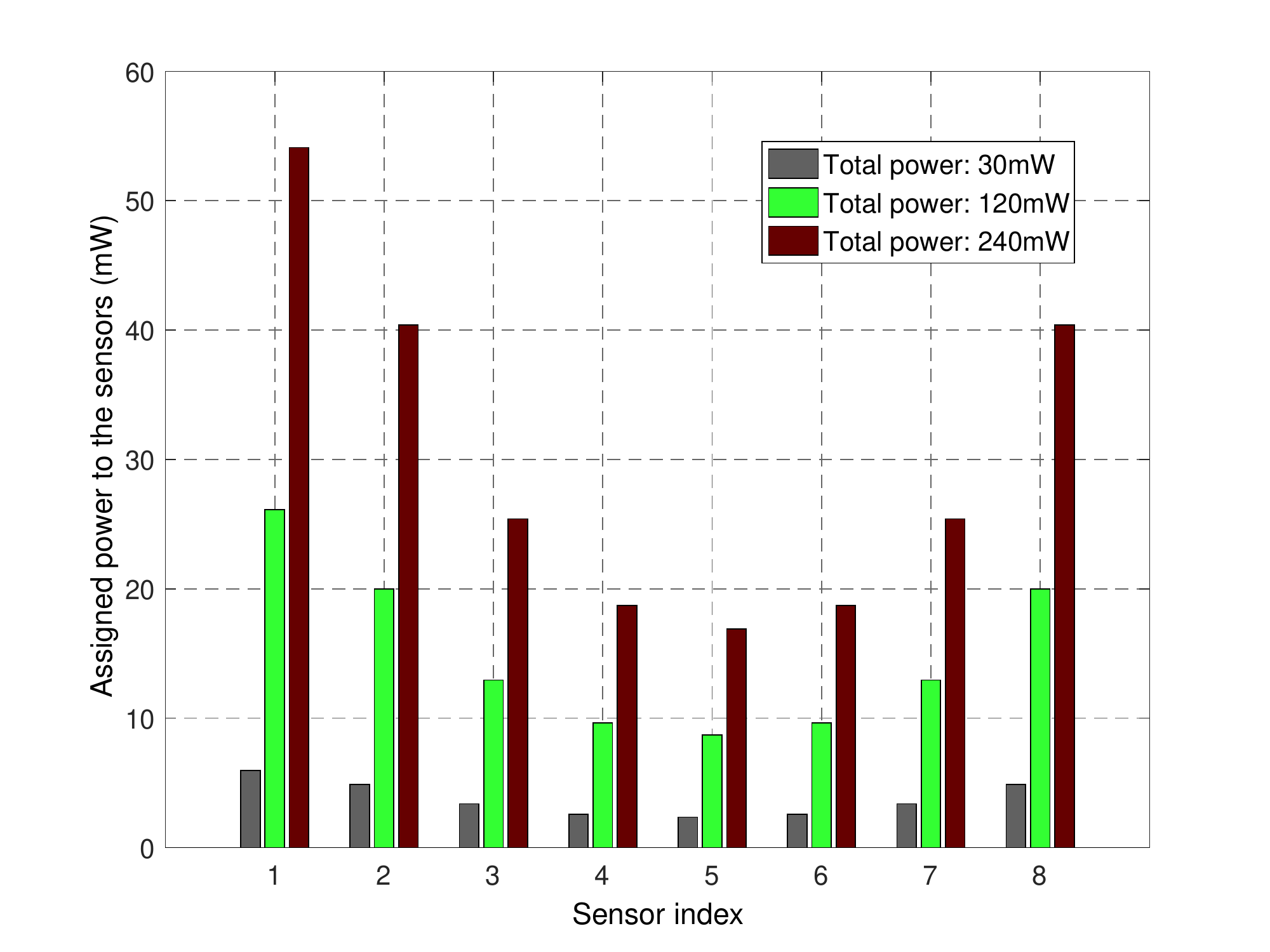}
\label{fig:coh_PAC_case1_01}
}
\subfigure[]{
\includegraphics[width=2.8in]{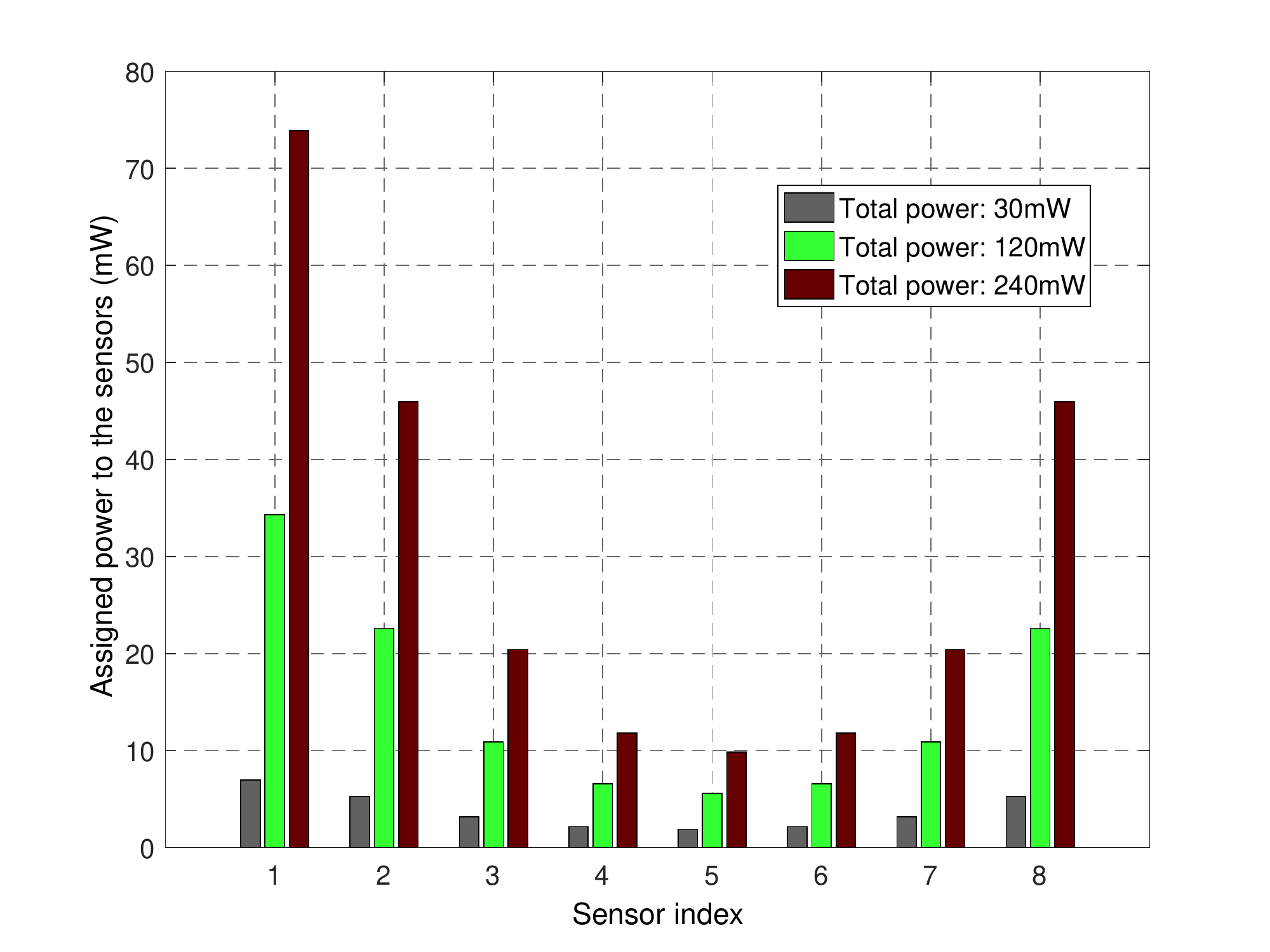}
\label{fig:coh_PAC_case1_09}
}
\subfigure[]{
\includegraphics[width=2.8in]{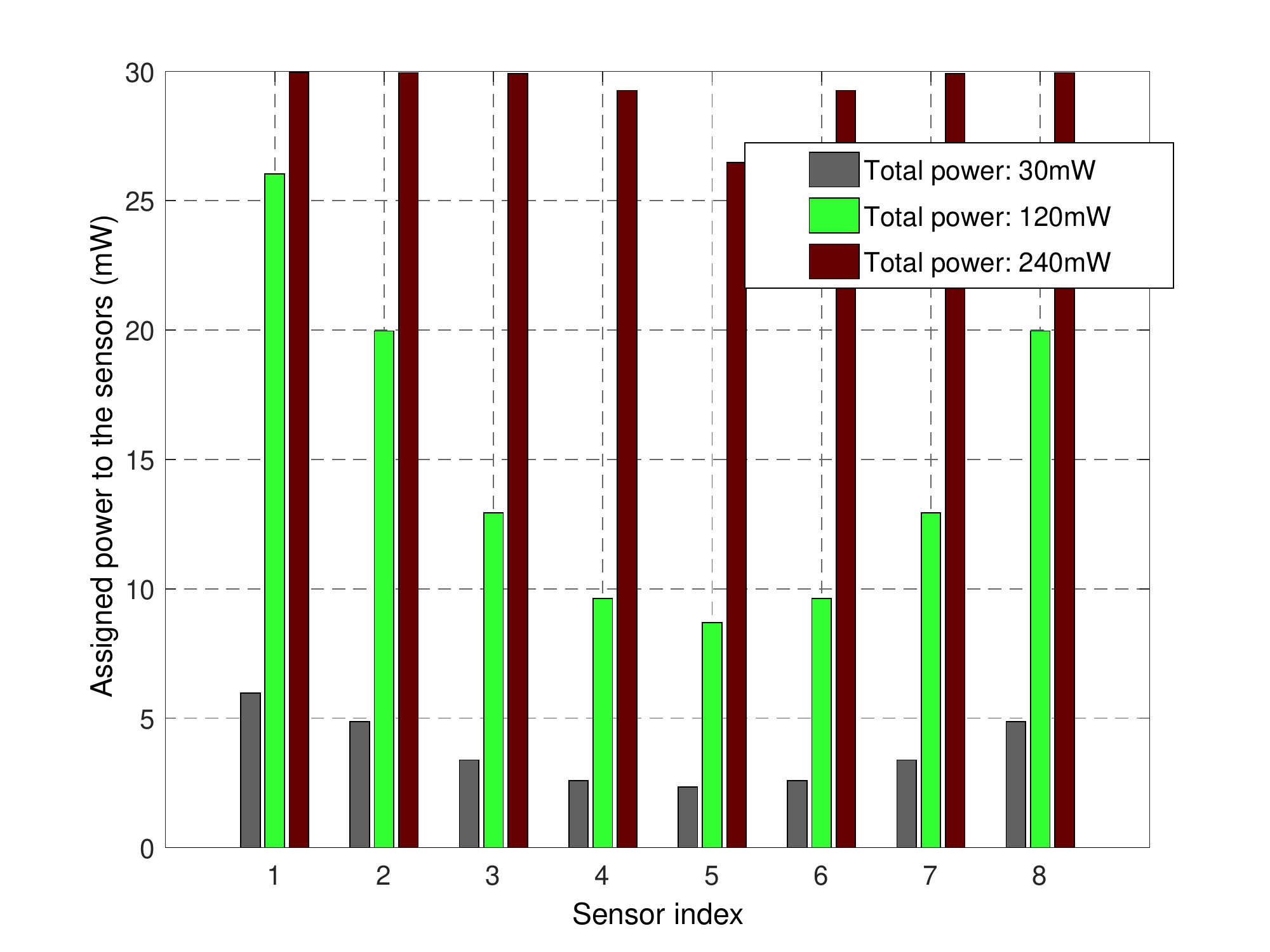}
\label{fig:coh_PAC_case2_01}
}
\subfigure[]{
\includegraphics[width=2.8in]{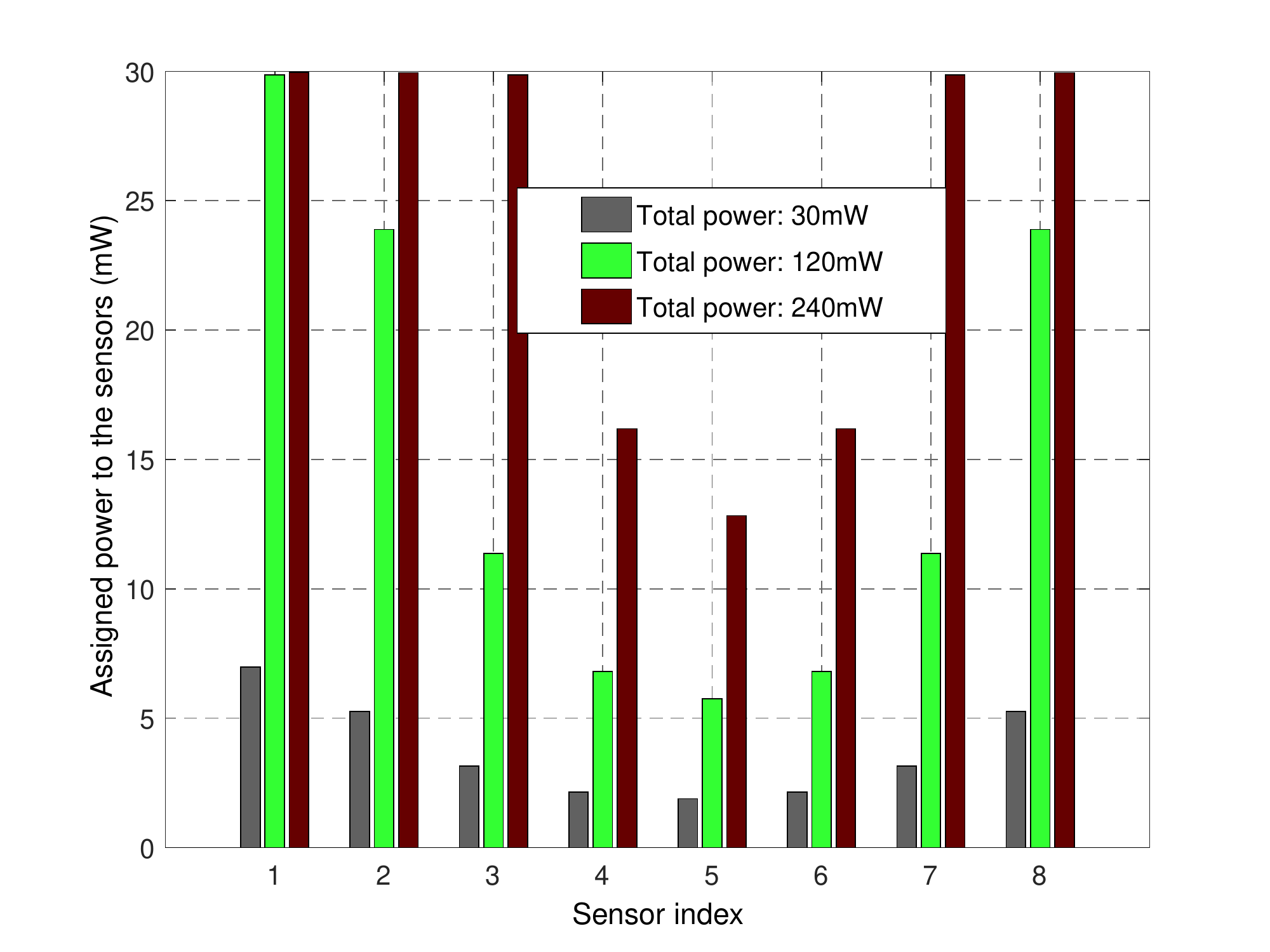}
\label{fig:coh_PAC_case2_09}
}
\subfigure[]{
\includegraphics[width=2.8in]{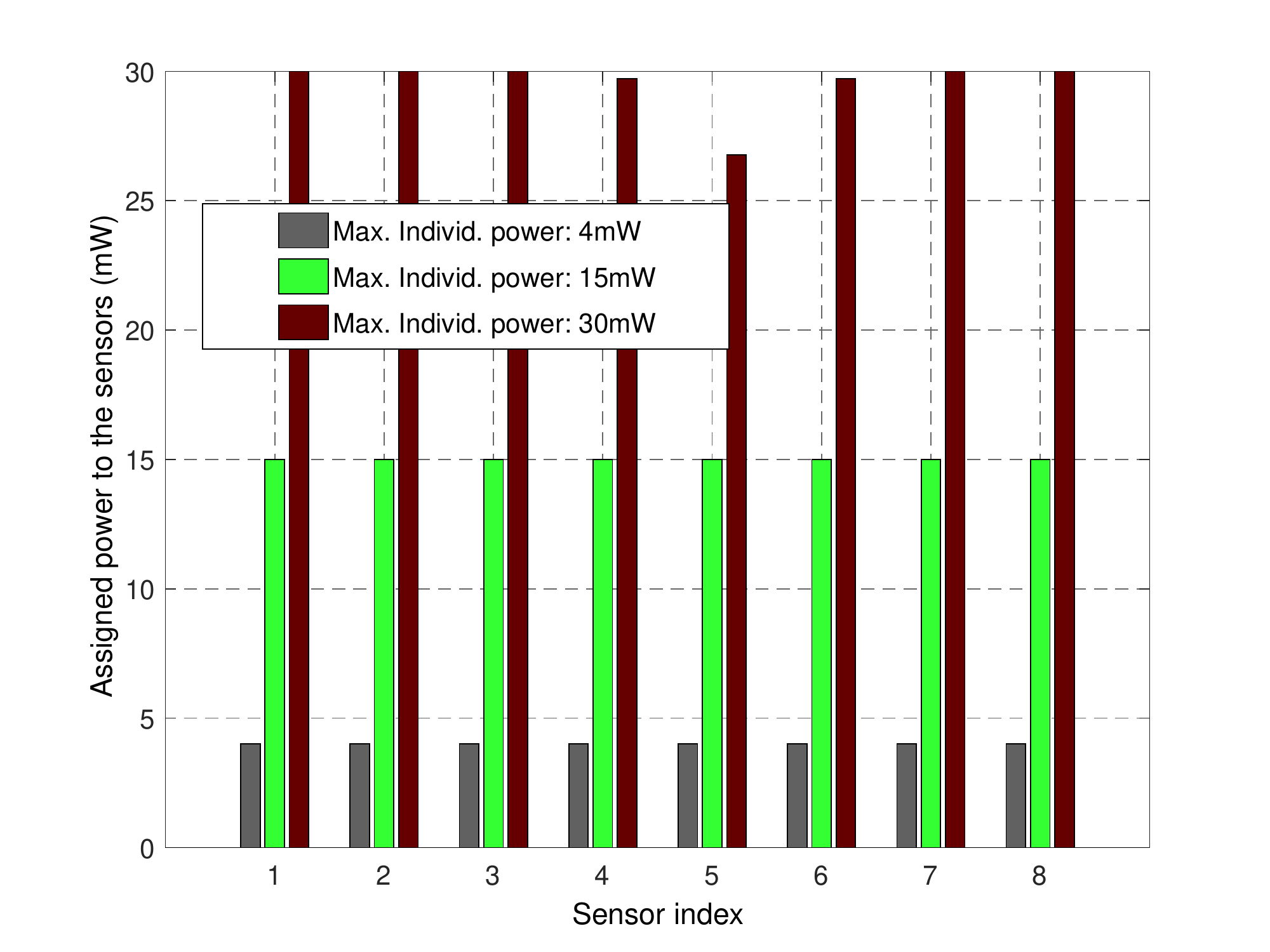}
\label{fig:coh_PAC_case3_01}
}
\subfigure[]{
\includegraphics[width=2.8in]{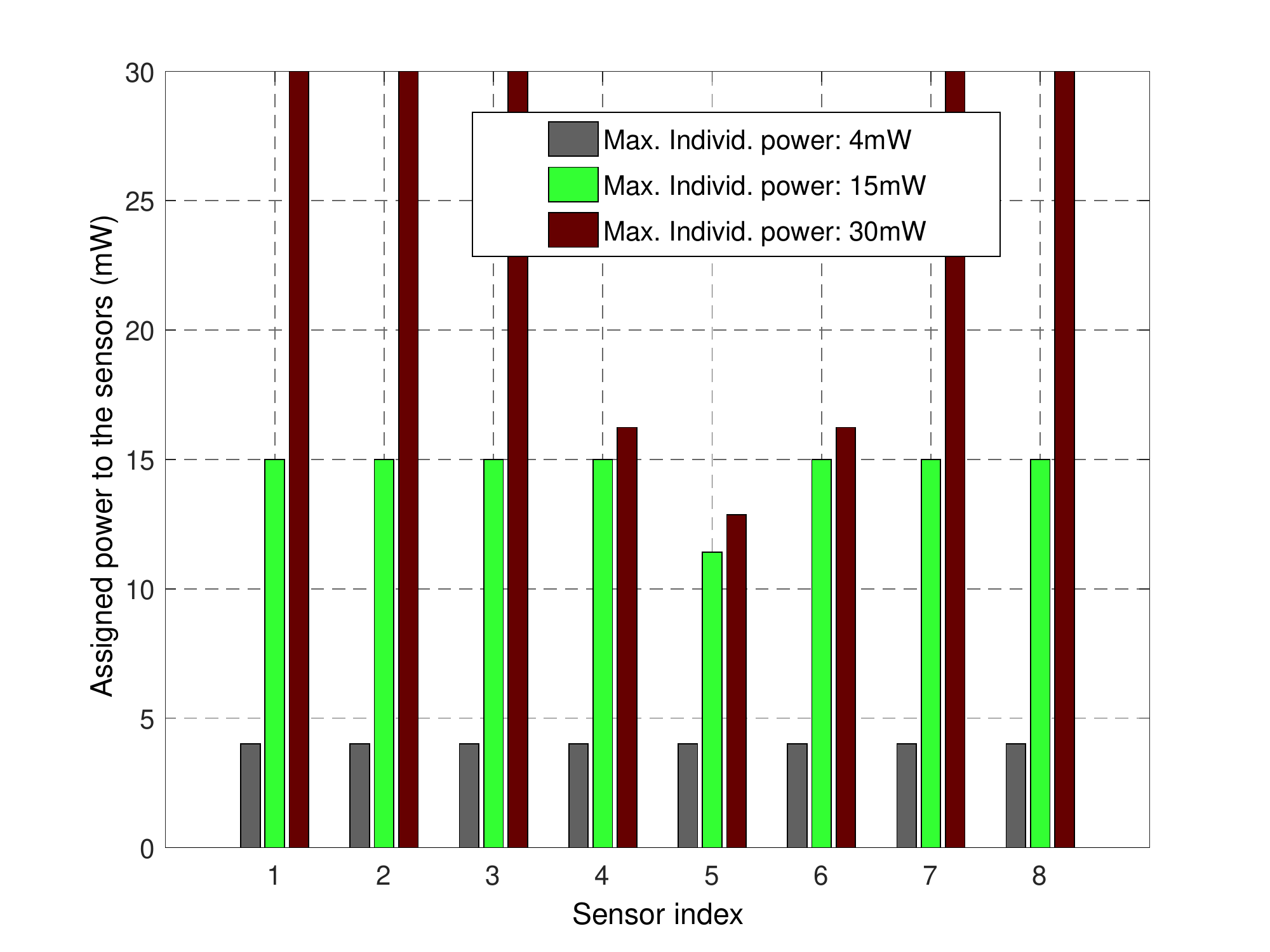}
\label{fig:coh_PAC_case3_09}
}
\caption{DPA in PAC with different $p_{d_k}$'s and identical pathloss: (a) Maximized MDC under TPC,  $\rho =0.1$; (b) Maximized MDC under TPC, $\rho =0.9$; (c) Maximized MDC under TIPC, $\rho =0.1$, $\bar{\cP} = 30\, \mW$, (d)  Maximized MDC under TIPC, $\rho =0.9$, $\bar{\cP} = 30\, \mW$; (e) Maximized MDC under IPC, $\rho =0.1$, (f) Maximized MDC under IPC, $\rho =0.9$.}
\label{fig:coh_PAC}
\end{figure}
%
%===============Trends of DPA when pathloss across sensors are different===================
%
\begin{figure}[b]
\centering
\subfigure[]{
\includegraphics[width=2.4in,height=1.6in]{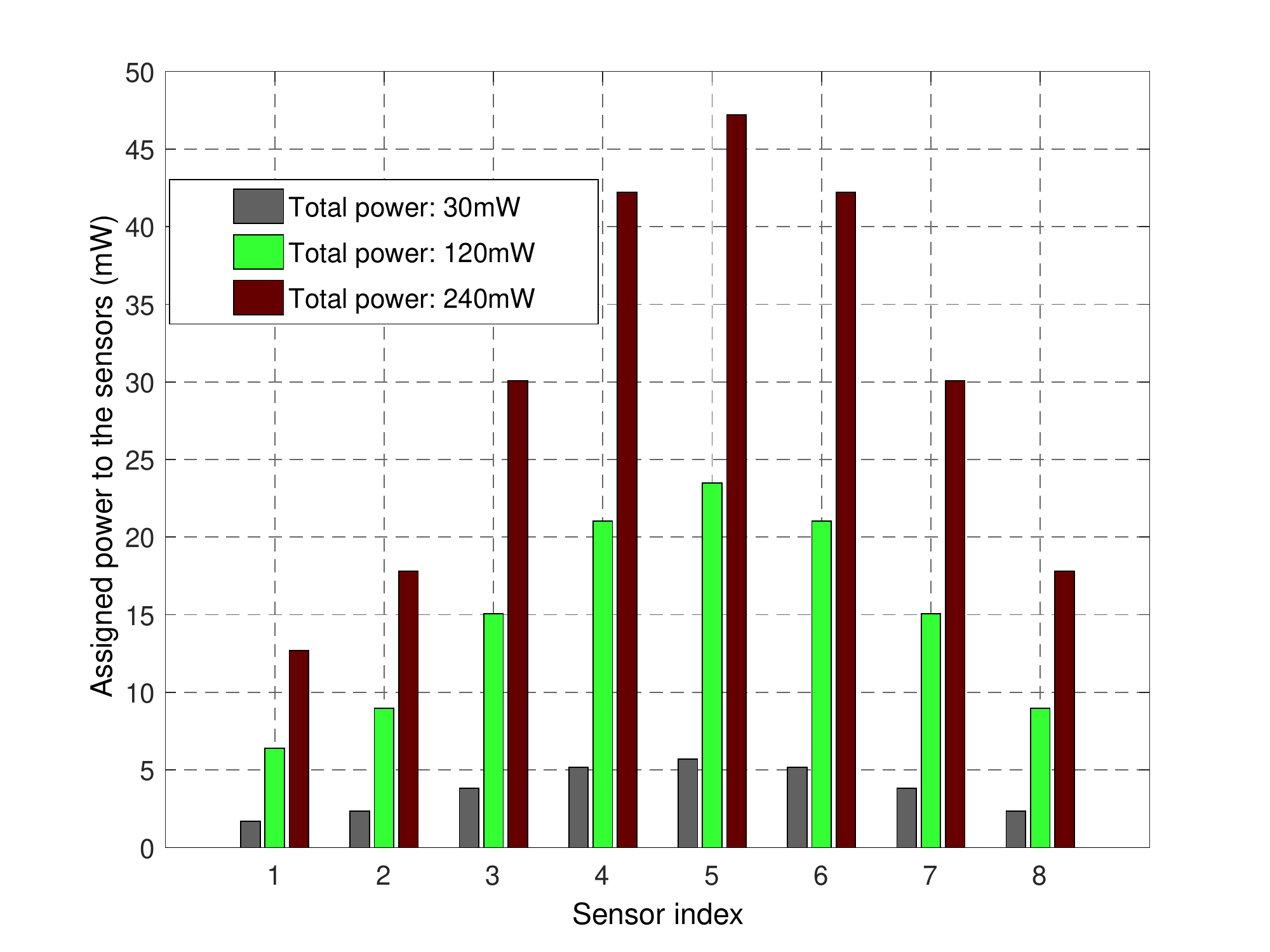}
\label{fig:coh_MAC_case1_01-pathloss2}
}
\subfigure[]{
\includegraphics[width=2.4in,height=1.6in]{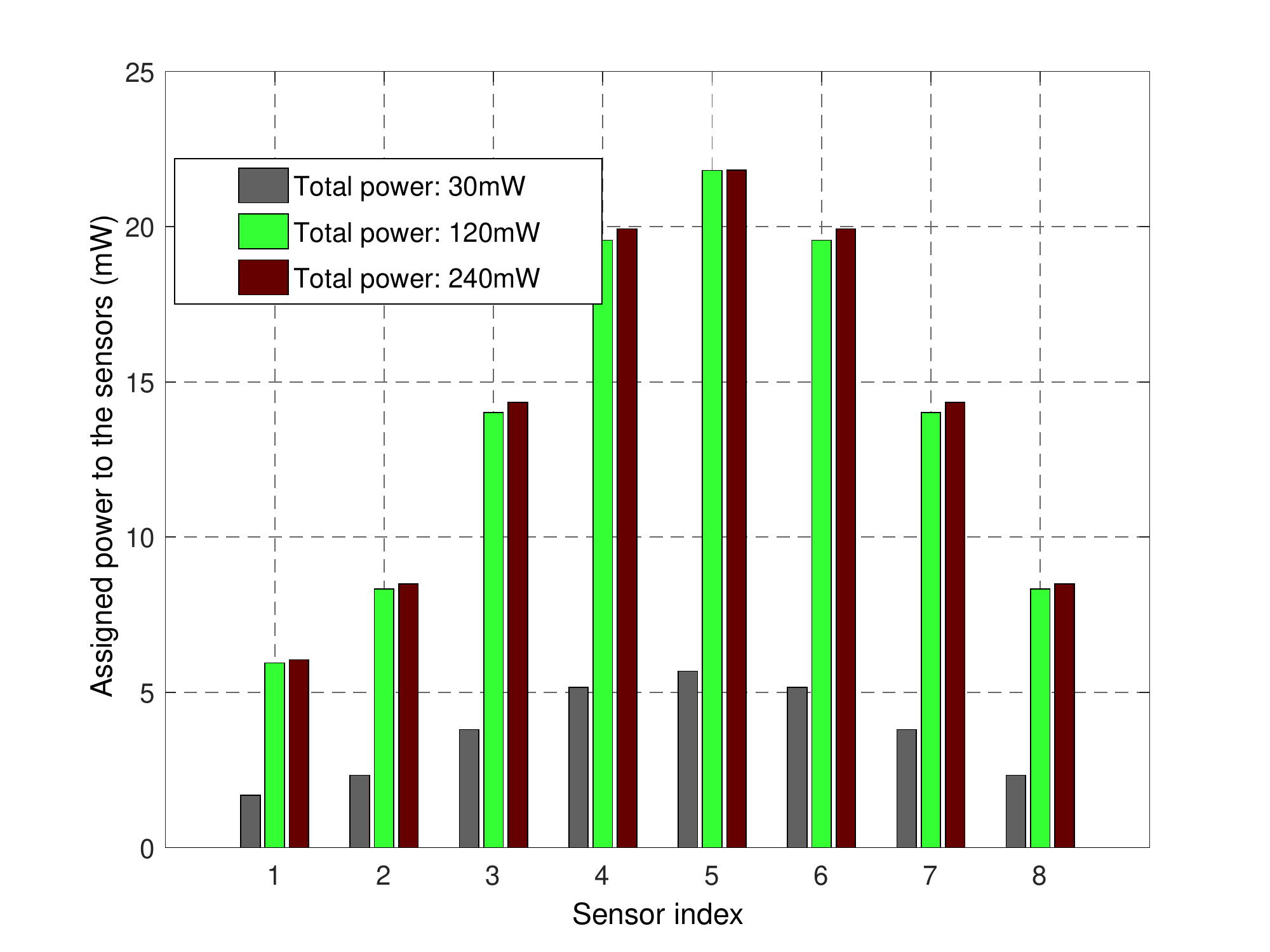}
\label{fig:coh_MAC_case2_01-pathloss2}
}
\subfigure[]{
\includegraphics[width=2.4in,height=1.6in]{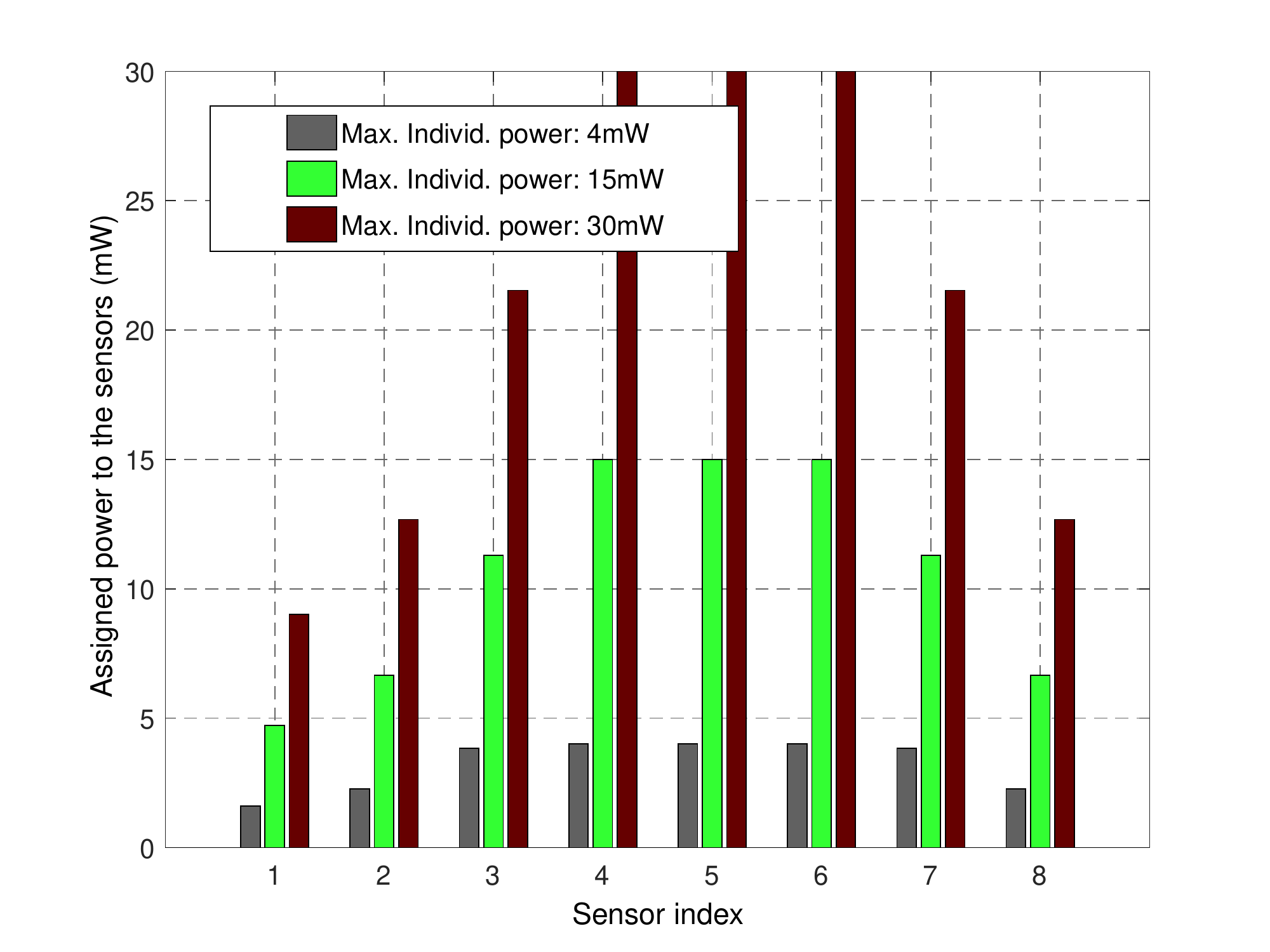}
\label{fig:coh_MAC_case3_01-pathloss2}
}
\caption{DPA in MAC with identical $p_{d_k}$'s, different pathloss and $\rho \!=\!0.1$: (a) Maximized MDC under TPC, (b) Maximized MDC under TIPC,  $\bar{\cP} \!=\! 30\, \mW$,  (c) Maximized MDC under IPC. }
\label{fig:coh_MAC-pathloss2}
\end{figure}
\begin{figure}[b]
\centering
\subfigure[]{
\includegraphics[width=2.4in,height=1.6in]{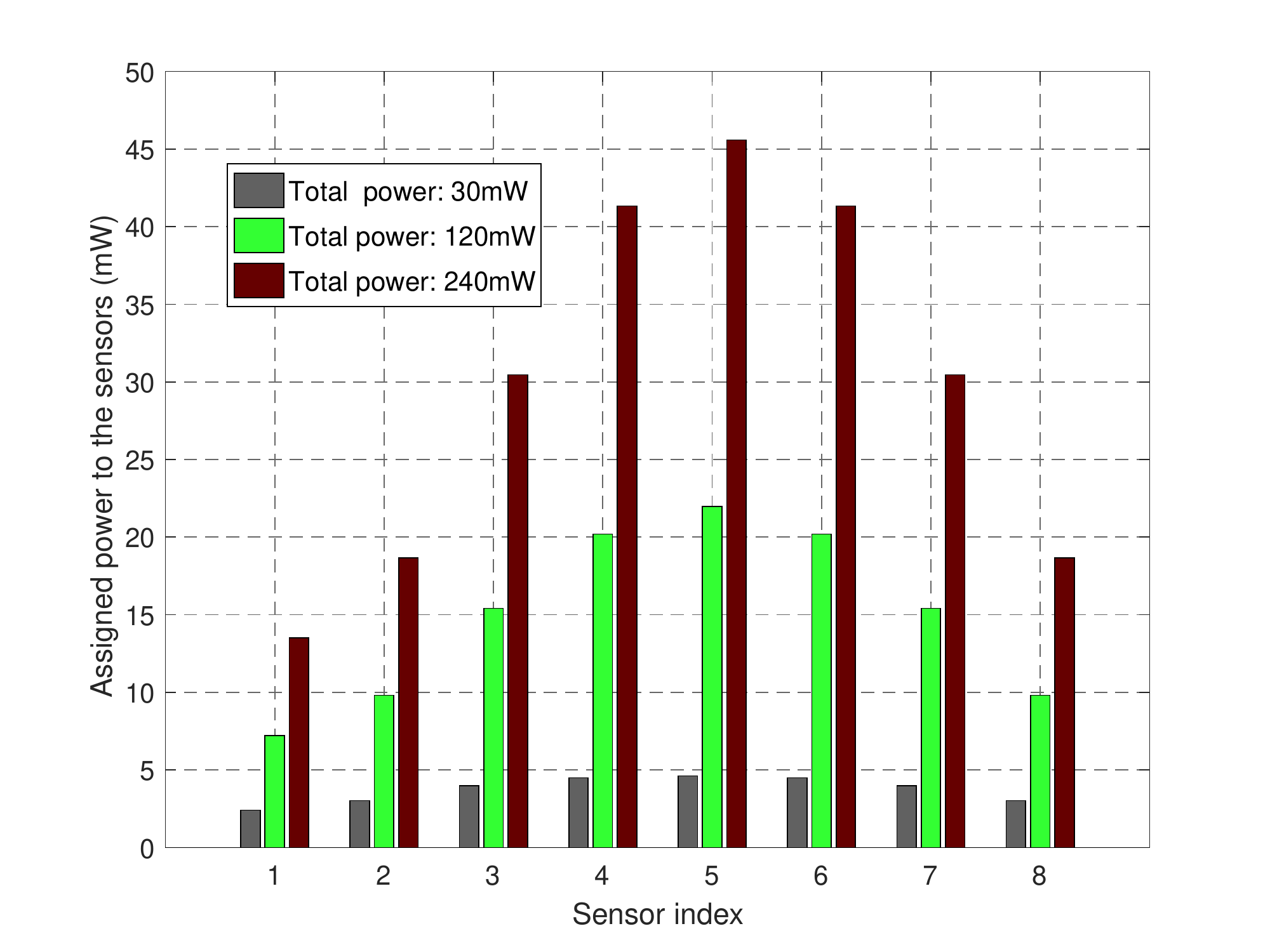}
\label{fig:coh_PAC_case1_01-pathloss2}
}
\subfigure[]{
\includegraphics[width=2.4in,height=1.6in]{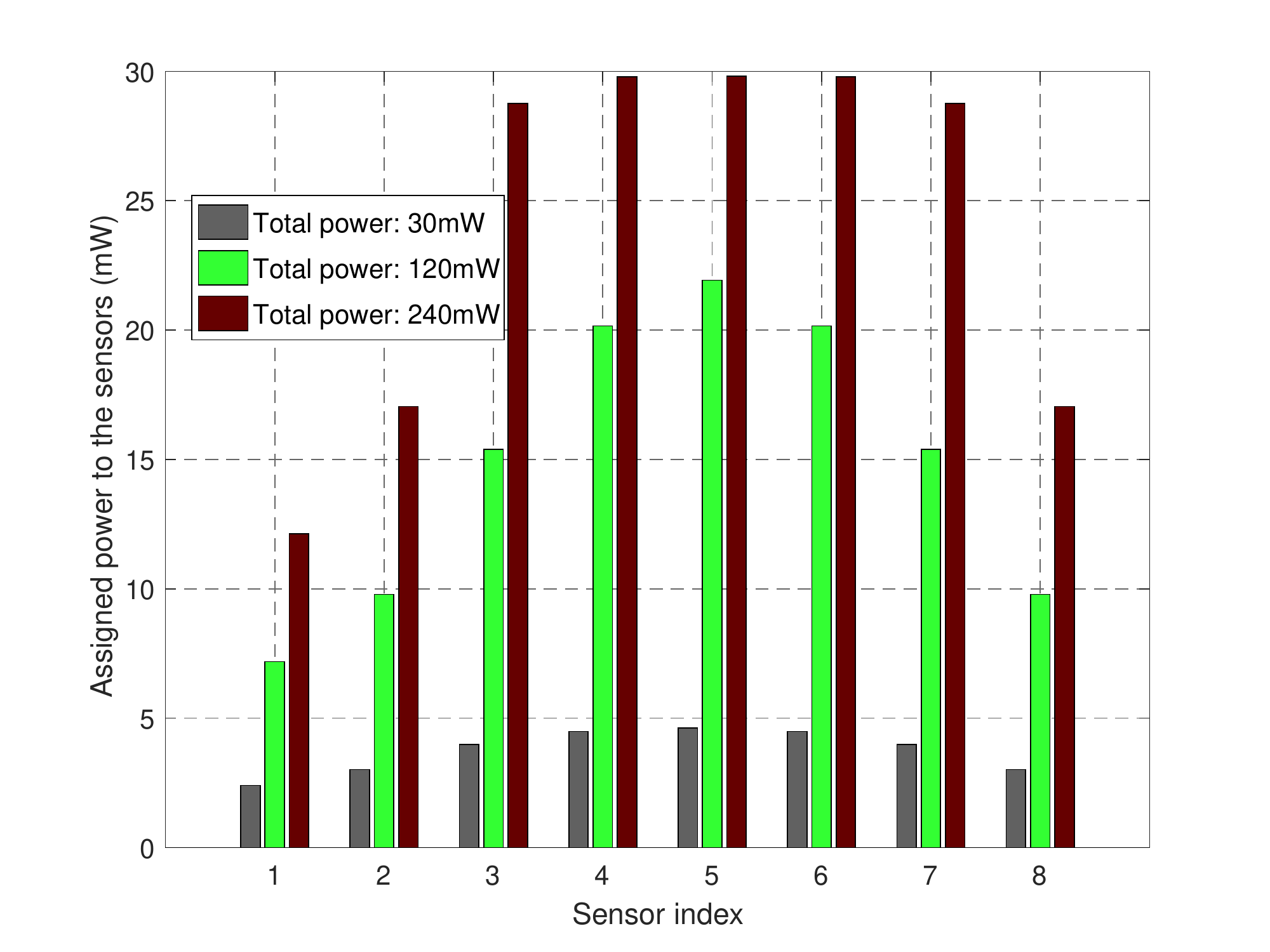}
\label{fig:coh_PAC_case2_01-pathloss2}
}
\subfigure[]{
\includegraphics[width=2.4in,height=1.6in]{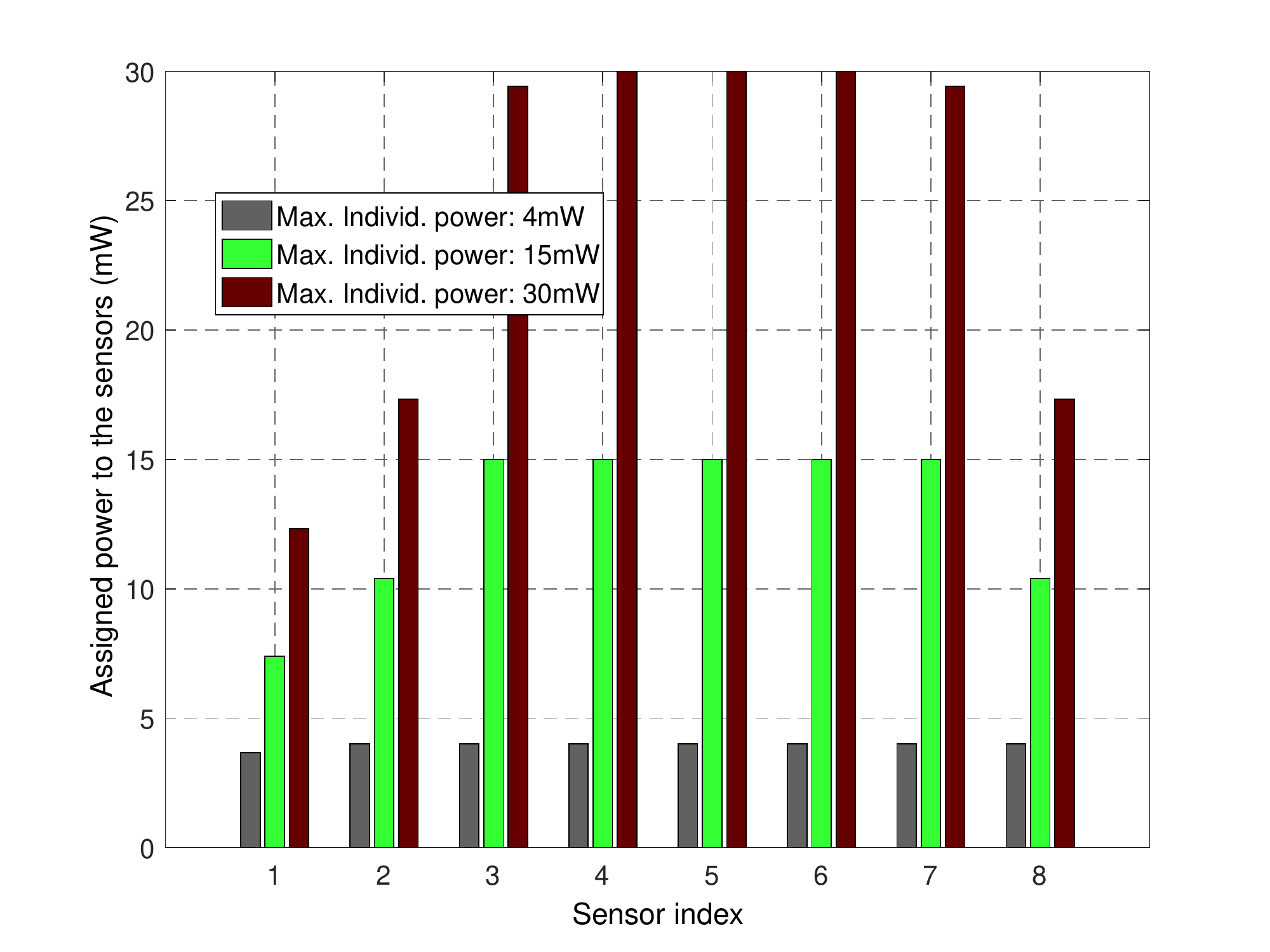}
\label{fig:coh_PAC_case3_01-pathloss2}
}
\caption{DPA in PAC with identical $p_{d_k}$'s, different pathloss and $\rho \!=\!0.1$: (a) Maximized MDC under TPC, (b) Maximized MDC under TIPC,  $\bar{\cP} \!=\! 30\, \mW$,  (c) Maximized MDC under IPC. }
\label{fig:coh_PAC-pathloss2}
\end{figure}
%
%===========================================
\fi
\end{document}